\documentclass[11pt]{article}
\usepackage[margin=1in]{geometry}
\usepackage{amsmath,amsthm,amssymb,enumerate, graphicx,subcaption}
\usepackage{setspace}
\usepackage{thmtools}
\usepackage{thm-restate}
\usepackage[hidelinks]{hyperref}
\usepackage{cleveref,color}
\usepackage[round]{natbib}
\usepackage[mathscr]{euscript}
\usepackage[shortlabels]{enumitem}
\usepackage{prodint}
\usepackage{graphics}
\usepackage{hyperref}
\usepackage{changepage}
\usepackage{breqn}

\newif\ifproof
\prooffalse

\declaretheorem[name=Theorem]{thm}

\declaretheorem[name=Lemma]{lemma}

\declaretheorem[name=Definition]{defn}

\DeclareMathOperator*{\indist}{\stackrel{d}{\longrightarrow}}

\DeclareMathOperator*{\fasterthan}{o_P}
\DeclareMathOperator*{\fasterthandet}{o}
\DeclareMathOperator*{\boundeddet}{O}

\newcommand{\s}[1]{\mathscr{#1}}
\renewcommand{\d}[1]{\mathbb{#1}}

\newcommand{\n}[1]{\mathrm{#1}}
\newcommand{\pp}[1]{\Phi^- \circ \Phi{(#1)}}

\allowdisplaybreaks 

\title{Nonparametric inference under a monotone hazard ratio order}
\author{Yujian Wu \and Ted Westling}
\date{Department of Mathematics and Statistics \\
University of Massachusetts Amherst \\
\today}

\begin{document}

\maketitle

\begin{abstract}
   The ratio of the hazard functions of two populations or two strata of a single population plays an important role in time-to-event analysis. Cox regression is commonly used to estimate the hazard ratio under the assumption that it is constant in time, which is known as the proportional hazards assumption. However, this assumption is often violated in practice, and when it is violated, the parameter estimated by Cox regression is difficult to interpret. The hazard ratio can be estimated in a nonparametric manner using smoothing, but smoothing-based estimators are sensitive to the selection of tuning parameters, and it is often difficult to perform valid inference with such estimators. In some cases, it is known that the hazard ratio function is monotone. In this article, we demonstrate that monotonicity of the hazard ratio function defines an invariant stochastic order, and we study the properties of this order. Furthermore, we introduce an estimator of the hazard ratio function under a monotonicity constraint. We demonstrate that our estimator converges in distribution to a mean-zero limit, and we use this result to construct asymptotically valid confidence intervals. Finally, we conduct numerical studies to assess the finite-sample behavior of our estimator, and we use our methods to estimate the hazard ratio of progression-free survival in pulmonary adenocarcinoma patients treated with gefitinib or carboplatin-paclitaxel.
\end{abstract}

\doublespacing

\section{Introduction}

\subsection{Background and literature review}

Time-to-event data are commonplace in many scientific fields, including biomedicine, economics, and engineering. In many circumstances, interest focuses on comparing the distribution of the time it takes for some event to occur, known as the event time, in two populations. For instance, in the medical sciences, patients may be randomly assigned to treatment or control, and followed until an event of interest occurs, such as onset, recurrence, or cure of a disease. In this case, the two populations are patients randomized to treatment and patients randomized to control. While the methods discussed in this paper are applicable to any time-to-event data, we will use ``patients" to refer to the units in the population of interest for convenience.

In the analysis of time-to-event data, one common parameter of interest is the cumulative distribution function of the event time, or equivalently, its survival function. However, in many settings, the event time is not observed for all patients in the study because, for example, some patients may prematurely leave the study, or the event may not have occurred before the end of the study period. This is known as right-censoring of the event time. If the censoring process is independent of the event process, the Kaplan-Meier estimator \citep{kaplan1958nonparametric} is a consistent nonparametric estimator of the survival function of the event time.

The distribution and survival functions describe cumulative probabilities, but in some cases it is of interest to quantify the instantaneous rate of the event at a point in time among patients who have not yet experienced the event of interest. This is known as the hazard rate. When comparing the distributions of an event time in two populations, the ratio of the hazard rates, known as the hazard ratio, describes the relative event rates among patients who have not yet experienced the event in the two populations over time. Estimating the hazard rate or ratio is more difficult than estimating the survival function because the hazard rate and ratio concern events occurring in an infinitesimal window of time. However, estimation of the hazard ratio is made much simpler by assuming that it is constant in time, which is known as the proportional hazards assumption. When this assumption holds, Cox proportional hazards regression can be used to estimate the hazard ratio \citep{cox1972regression}. In this case, the hazard ratio for comparing two populations reduces to a single number. 
The hazard ratio estimated from a simple Cox regression comparing two populations has become one of the most important tools in the analysis of time-to-event data, and in some studies it is the only effect reported \citep{hernan2010hazards}. 

Despite the widespread use of Cox regression, the proportional hazards assumption underlying it is easily violated. For example, if a treatment only offers short-term benefits over control, then the hazard ratio is unlikely to be constant \citep{li2015statistical}. In addition, the proportional hazards assumption implies that the survival function of one group can be expressed as the survival function of the other group raised to a constant power. Hence, if the survival curves cross, then the proportional hazards assumption cannot hold (see, e.g.\ \citealp{klein2003survival}). The hazard ratio estimated by a Cox regression in a setting where the proportional hazards assumption is violated is approximately a weighted average of the hazard ratio function over time \citep{struthers1986misspecified}. However, the weighting function depends on the censoring pattern in the study, which complicates the interpretation of the parameter estimated by the Cox model in such a misspecificed model \citep{o2008proportional, whitney2019comment}

When the proportional hazards assumption is violated, estimating the hazard ratio function is more difficult. One simple approach is to estimate the hazard ratio using the ratio of estimators of the individual hazard rate functions. For example, if correctly specified parametric models for the distributions are available, the hazard rates in the two distributions can be estimated using maximum likelihood estimation \citep{kalbfleisch2011statistical}. Alternatively, nonparametric methods for estimating hazard functions based on smoothing have also been proposed \citep{anderson1980smooth, Mller1994HazardRE, rebora2014bshazard}. However, estimators based on smoothing are often sensitive to the selection of certain tuning parameters, such as bandwidths, kernel functions, or the number of knots in a spline function. In addition, obtaining valid inference using a smoothing-based estimator can be challenging due to bias in the asymptotic distribution of the estimator (see, e.g.\  \citealp{wasserman2013all} and \citealp{calonico2018effect}).  

In some cases, it may be known that the hazard ratio is monotone as a function of time. In general, the hazard ratio can be expected to be monotone when the relative rate of events in the two groups increases or decreases over time. For example, if the effectiveness of  a treatment wanes over time, then the hazard ratio between treated and  placebo groups of a randomized trial may be expected to be monotone non-decreasing \citep{durham1998estimation}. Similarly, harmful exposures can result in a monotone non-decreasing hazard ratio between the exposed and unexposed groups \citep{sekula2013comprehensive}. We discuss the motivation and application of monotone hazard ratios more in Section~\ref{sec:monotone hr ordering}.

We are only aware of a few studies concerning monotonicity of the hazard ratio function. \cite{gill1987simple} and \cite{deshpande1995testing} proposed tests of the proportional hazards assumption against the non-decreasing hazard ratio alternative. \cite{kim2011bayesian} proposed an estimator of a monotone hazard ratio function using a nonparametric Bayesian approach, which we discuss further in Section~\ref{sec:estimation}.

\subsection{Contribution and organization of the article}

In this article, we study the situation in which the hazard ratio between two populations is known to be non-decreasing in time. First, we define a new stochastic order called the monotone hazard ratio order, demonstrate that it is an invariant stochastic order in the sense of \cite{lehmann1992invariant}, and study the properties of this novel stochastic order. As we will discuss more below, this is important because it gives stability to the monotonicity assumption, and because it connects our new order to the existing literature on stochastic orders. Second, we propose a novel estimator of a hazard ratio function under a monotonicity constraint in the presence of independent right-censoring. Finally, we derive the large-sample properties of our estimator, including convergence in distribution of our estimator at the rate $n^{-1/3}$ to a mean-zero limit, and use this result to construct asymptotically valid pointwise confidence intervals for the hazard ratio function. To the best of our knowledge, we are the first to study the stochastic order defined by monotonicity of the hazard ratio function, and we are also the first to produce asymptotically valid confidence intervals for a monotone hazard ratio function.

The paper proceeds as follows. In Section~\ref{sec:monotone hr ordering}, we define the monotone hazard ratio order and establish properties of this order. In Section~\ref{sec:estimation}, we introduce our nonparametric estimator of a monotone hazard ratio function, establish asymptotic theory of our estimator, and use this theory to construct confidence intervals. In Section~\ref{sec:simulation}, we present numerical studies evaluating the finite-sample performance of our method. Finally, in Section~\ref{sec:application}, we use our method to estimate the hazard ratio function comparing the length of progression-free survival of pulmonary adenocarcinoma patients treated with gefitinib or carboplatin-paclitaxel. Proofs of all theorems can be found in Supplementary Material.

\subsection{Notation}

For a function $H$ on a domain $\s{I} \subseteq \d{R}$ to the extended real line $\bar{\d{R}}$, we let $\bar{H} := 1-H$. If $H$ possesses limits from the left, then we let $H_- := x \mapsto H(u-) := \lim_{u\uparrow x} H(u)$ be the left-continuous version of $H$, and if $H$ possesses limits from the right, then we let $H_+ := x \mapsto H(x+) := \lim_{u\downarrow x} H(u)$ be the right-continuous version of $H$. We set $\Delta H := H_+ - H_-$. If $H$ is left-differentiable at $x \in \s{I}$, we denote by $\partial_- H(x)$ the left derivative of $H$ at $x$. We also denote the image of $H$ by $\n{Im}(H) := \{u \in \d{R} : H(x) = u$ for some $x \in \s{I}\}$. If $H$ is non-decreasing, we define the support of $H$ as $\n{Supp}(H) := \{x \in \s{I} : H(u) < H(v)$ for all $u <  x < v\}$.  We define the \emph{greatest convex minorant} (GCM) of $H$ on $\s{I}$, denoted $\n{GCM}_{\s{I}}(H) : \s{I} \to\bar{\d{R}}$,  as the pointwise supremum of all convex functions on $\s{I}$ bounded above by $H$. We say that $H$ is \emph{monotone on} $\s{A} \subseteq \s{I}$ if $H(x) \leq H(y)$ for all $x < y$ with $x, y \in \s{A}$, and similarly we say that $H$ is \emph{convex on} $\s{A}$ if $H(tx + (1-t)y) \leq tH(x) + (1-t)H(y)$ for all $x,y \in \s{A}$ and $t \in [0,1]$ such that $tx + (1-t)y \in \s{A}$ as well. We set $H^-(u) := \inf\{t \leq u : H(t) \geq H(u)\}$ as the generalized inverse function corresponding to $H$. The properties of such functions when $H$ is a distribution function (in which case $H^-$ is its quantile function) are summarized in Chapter~21 of \cite{van2000asymptotic}. All integrals should be interpreted as Riemann-Stieltjes integrals, and $\int_0^t := \int_{(0,t]}$ by default.

\section{Monotone hazard ratio ordering \label{sec:monotone hr ordering}}

\subsection{Definition of the monotone hazard ratio order}

We now introduce and motivate the monotone hazard ratio order. We let $S$ and $T$ be nonnegative random variables, and we let $F_S$, $\bar{F}_S$, $F_T$, and $\bar{F}_T$ be the distribution and survival functions corresponding to $S$ and $T$, respectively. If $S$ and $T$ are absolutely continuous with density functions $f_S = F'_S$ and $f_T = F'_T$, then $\lambda_S := f_S / \bar{F}_S$ and $\lambda_T := f_T / \bar{F}_T$ are the hazard functions corresponding to $S$ and $T$, respectively. In this case, we say $S \geq_{MHR} T$ if  $t \mapsto \theta(t) := \lambda_S(t) / \lambda_T(t)$ is non-decreasing for $t$ such that $f_T(t) > 0$ or $f_S(t) > 0$. On the other hand, if $S$ and $T$ are fully discrete random variables with support contained on a finite or countably infinite set $\{t_1 < t_2 < \cdots \}$, then $\lambda_S(t_j) := f_S(t_j) / \bar{F}_S(t_{j-1})$ and $\lambda_T(t_j) := f_T(t_j) / \bar{F}_T(t_{j-1})$ are the corresponding hazard functions, where $f_S(t) := P(S =t)$ and $f_T(t) := P(T = t)$ are the corresponding mass functions (and where $t_0 := -\infty$). In this case, we say $S \geq_{MHR} T$ if $t \mapsto \theta(t) := \lambda_S(t) / \lambda_T(t)$ is non-decreasing for all $t \in \{t_1, t_2, \dots\}$ such that $f_T(t) > 0$ or $f_S(t) > 0$.

We define $\geq_{MHR}$ in such a way that encompasses both the above cases, as well as more complicated cases where $S$ and $T$ may be mixed discrete-continuous random variables. We let $\mu$ be any sigma-finite measure dominating both $F_S$ and $F_T$, and we define $f_S := dF_S / d\mu$ and $f_T := dF_T / d\mu$. We then define the hazard functions relative to $\mu$ as $\lambda_S := f_S / \bar{F}_{S,-}$ on the support of $f_S$, and 0 otherwise, and similarly for $\lambda_T$. The hazard ratio function $\theta : \n{Supp}(F_S) \cup \n{Supp}(F_T) \to [0, \infty]$  is then defined as $\theta := \lambda_S / \lambda_T$. We note that $\theta$ does not depend on the choice of dominating measure $\mu$, that $\theta = 0$ on $\n{Supp}(F_T) \backslash \n{Supp}(F_S)$, and that $\theta = +\infty$ on $\n{Supp}(F_S) \backslash \n{Supp}(F_T)$. We then have the following general definition of the monotone hazard ratio relation.
\begin{defn}
We say that $S \geq_{MHR} T$ if $\theta = \lambda_S / \lambda_T$ is non-decreasing on $\n{Supp}(F_S) \cup \n{Supp}(F_T)$.
\end{defn}
When both $S$ and $T$ are dominated by Lebesgue measure, we recover the first case discussed above, and when they are both dominated by counting measure on the countable set $\{t_1 < t_2 < \cdots \}$, then we recover the second case. 

Monotone hazard ratios abound in the literature because monotonicity of hazard ratio function can be expected to hold in several general situations. First, if $S$ is the time to an adverse event under treatment and $T$ is the same under control, we can expect $S \geq_{MHR} T$ if the protective effect of the treatment on those who have not yet experienced it wanes over time. There are many examples of such treatments, including vaccines \citep{durham1998estimation} and blood transfusion \citep{holcomb2013prospective}. Second, if $S$ is the time to an adverse event under control, and $T$ is the same under exposure to a condition with short-term toxic effects, then we may again expect that $S \geq_{MHR} T$. Drug overdose is an example of such a toxic exposure \citep{hernandez2018exposure}. We note that the individual hazard functions of $S$ and $T$ may not be monotone in the above cases. For instance, there may be underlying time trends (e.g., weekly, monthly, or seasonal trends) unrelated to treatment that induce non-monotonic trends in the hazards. If these trends influence the hazards of $S$ and $T$ equally, then the hazard ratio may still be expected to be monotone.

The statistical model induced by the monotone hazard ratio order is a generalization of the popular proportional hazards model with a time trend, where the time trend is allowed to be any monotone function. Choosing a specific time trend for a proportional hazards model can be difficult, and if the time trend is chosen based on the data, obtaining valid inference for the regression coefficient is challenging \citep{desquilbet2005time}. Hence, the flexibility in permitting any monotone time trend is appealing because it avoids the need to choose a specific trend. 

We will see in Section~\ref{sec:estimation} that when it is known that the hazard ratio is monotone, this knowledge can be exploited to obtain a simple nonparametric estimator of the hazard ratio function and asymptotically valid pointwise inference. Furthermore, the estimator and inferential procedure avoid estimating or modelling the individual hazard functions directly, in the same spirit as the proportional hazards estimator, which yields improved robustness over methods that estimate the hazard functions. This will be explored more in numerical studies in Section~\ref{sec:simulation}.


\subsection{Properties of the monotone hazard ratio order}

We now establish several important properties of the monotone hazard ratio order. First, we show that the relation defined above is an invariant stochastic order in the sense of \cite{lehmann1992invariant}. Intuitively, stochastic orders are ways of defining what it means for one probability distribution to be ``larger" than another. Specifically, a stochastic order $\geq_S$ is a relation on the space of probability distributions on some measurable space satisfying the conditions of a preorder: for any probability distributions $F$, $G$, and $H$ on the space, (1) $F \geq_S F$, and (2) $G \geq_S F$ and $H \geq_S G$ implies that $H \geq_S F$. We will be focused on distributions on the reals. In this case, a stochastic order is \emph{invariant under monotone transformations}, or simply \emph{invariant}, if $G \geq_S F$ implies $G \circ \psi^{-1}  \geq_S F \circ \psi^{-1}$ for any strictly increasing continuous function $\psi : \d{R} \to \d{R}$ with $\lim_{x \to -\infty} \psi(x) = -\infty$ and $\lim_{x \to \infty} \psi(x) = \infty$. For two real-valued random variables $S$ and $T$ with distribution functions $F_S$ and $F_T$, we say $S \geq_{S} T$ for a stochastic order $\geq_{S}$ if $F_S \geq_{S} F_T$. We now show that these properties hold for the monotone hazard ratio order defined above.

\begin{thm}\label{thm:invariant}
(1) For any random variable $S$, $S \geq_{MHR} S$; (2) for any $S$, $T$, and $U$ such that $S \geq_{MHR} T$ and $T \geq_{MHR} U$, it holds that $S \geq_{MHR} U$; and (3) for any strictly increasing continuous function $\psi$ with $\lim_{x \to -\infty} \psi(x) = -\infty$ and $\lim_{x \to +\infty} \psi(x) = +\infty$, $S \geq_{MHR} T$ implies that $\psi(S) \geq_{MHR} \psi(T)$.
\end{thm}
The fact that the monotone hazard ratio forms a stochastic order is important due to the stability it provides when comparing the hazard ratios of multiple event times. The fact that it is invariant to monotone transformations is especially important because it means that the order is independent of time scale. We also note that $S \geq_{MHR} T$ and $T \geq_{MHR} S$ implies that the hazard ratio is constant, but does not imply that $S = T$ in distribution. Hence, the monotone hazard ratio order is not antisymmetric, and therefore does not induce a partial order.

We now provide two characterizations of the monotone hazard ratio order in the special case where $F_S \ll F_T$, i.e.\ $F_S$ is dominated by $F_T$.  We define $\Lambda_S(t) := \int_0^t F_S(du) / \bar{F}_{S_-}(u)$ and $\Lambda_T(t) :=\int_0^t F_T(du) / \bar{F}_{T_-}(u)$ as the cumulative hazard functions corresponding to $S$ and $T$, respectively, and we note that if $F_S \ll F_T$, then $\Lambda_S \ll \Lambda_T$, and $\theta = d\Lambda_S / d\Lambda_T$. We also define $R = F_S \circ F_T^-$ as the \emph{ordinal dominance curve} corresponding to the distributions of $S$ and $T$. \cite{lehmann1992invariant} demonstrated that all invariant stochastic orders are equivalent to a pre-order on the space of ordinal dominance curves that is closed under composition. In the next result, we provide two characterizations of the monotone hazard ratio order: one in terms of the ordinal dominance curve, and a second in terms of the cumulative hazard functions $\Lambda_S$ and $\Lambda_T$.
\begin{thm}\label{thm:mhr_characterization}
(a) If $F_S \ll F_T$ and $\theta$ is continuous, then the following are equivalent:
\begin{enumerate}
    \item  $S \geq_{MHR} T$;
    \item $u \mapsto  \int_{[0,u)} (1-v) / \bar{R}(v)\, dR_+(v) = \int_{-\infty}^{F_T^-(u)} (d\Lambda_S / d\Lambda_T)(t) \, dF_T(t)$ is convex on $\n{Im}(F_T)$;
    \item $u \mapsto \Lambda_S \circ \Lambda_T^-(u)$ is convex on $\n{Im}(\Lambda_T)$.
\end{enumerate}
(b) If $F_S \ll F_T$, $\theta$ is continuous, and $S \geq_{MHR} T$, then $\theta(t) = \partial_- \n{GCM}_{I}(\Lambda_S \circ \Lambda_T^-) \circ \Lambda_T(t)$ for any $t \in \n{Supp}(F_T)$, where $I$ is the smallest closed interval containing $\n{Im}(\Lambda_T)$.
\end{thm}

The assumption that $F_S \ll F_T$ is important for the characterizations in Theorem~\ref{thm:mhr_characterization}. For example, if $F_S$ is the uniform distribution on $[0,1.5]$ and $F_T$ is the Bernoulli distribution with probability $1/2$, then $F_S$ is not dominated by $F_T$, but $\Lambda_S \circ \Lambda_T^-$ is convex on $\n{Im}(\Lambda_T)$ and $S \ngeq_{MHR} T$. This is similar to a counterexample provided in \cite{mosching2020estimation} for the likelihood ratio order. If treatment or exposure does not change the set of possible event times, which is the case in many real-world situations, then $F_S \ll F_T$ can be expected to hold.

The characterization of the monotone hazard ratio order in terms of the ordinal dominance curve provided in Theorem~\ref{thm:mhr_characterization} is somewhat more complicated than the characterization of the other three common invariant stochastic orders discussed below. This is due to the complexity of the general relationship between a hazard function and the corresponding distribution function. In the case of absolutely continuous $F_S$ and $F_T$, the characterization in terms of the ordinal dominance curve can be stated somewhat simpler. In particular, $\int_0^u (1-v) / \bar{R}(v) \, dR_+(v)$ is convex if and only if $v \mapsto (1-v) R'(v) / \bar{R}(v)$ is monotone, which holds if and only if $t \mapsto e^{-t} R'(1 - e^{-t}) / \bar{R}(1-e^{-t})$ is monotone. Then, since $\int_0^t e^{-s} R'(1 - e^{-s}) / \bar{R}(1-e^{-s}) \, ds = -\log \bar{R}(1 - e^{-t})$, in the absolutely continuous case the monotone hazard ratio order is equivalent to $t\mapsto R(1-e^{-t})$ being log-convex on $t \in [0, \infty)$. In this case, it is necessary but not sufficient that $R$ be log-convex.

The relationship between the hazard and cumulative hazard functions is analogous to that between a density and distribution function. Hence, the characterization of the monotone hazard ratio order in terms of the cumulative hazard functions provided in  Theorem~\ref{thm:mhr_characterization}  parallels the characterization of a likelihood ratio order in terms of the distribution functions \citep{westling2019nonparametric, mosching2020estimation}. We will see in Section~\ref{sec:estimation} that part (b) of Theorem~\ref{thm:mhr_characterization} suggests a natural estimator of $\theta$.

Theorem~\ref{thm:mhr_characterization} can also be used to informally assess the plausibility of the monotone hazard ratio order given data. We note that $\Lambda_S \circ \Lambda_T^-$ is convex on $\n{Im}(\Lambda_T)$ if and only if the parametrized curve $\{(\Lambda_T(t), \Lambda_S(t)) : t \in \n{Supp}(F_T)\}$ is convex on $\n{Supp}(F_T)$. Hence, if $\Lambda_{S,n}$ and $\Lambda_{T,n}$ are consistent estimators of $\Lambda_S$ and $\Lambda_T$, respectively, then $S \geq_{MRH} T$ if and only if $\{(\Lambda_{T,n}(t), \Lambda_{S,n}(t)) : t \in \n{Supp}(\Lambda_{T,n})\}$ is consistent for a convex function. Hence, comparing this curve to its GCM gives an informal graphical check of the monotone hazard ratio order. This same procedure was proposed by \cite{gill1987simple}.

\subsection{Relationship to other stochastic orders}

A variety of stochastic orders have been studied;  \cite{shaked2007stochastic} contains detailed results and discussion. We briefly review three of the most common stochastic orders used in the context of univariate time-to-event analysis. The \emph{usual} or \emph{uniform} stochastic order is defined as $G \geq_{ST} F$ if $G(t) \leq F(t)$ for all $t \in \d{R}$, where $F$ and $G$ are cumulative distribution functions on $\d{R}$. Estimators under the usual stochastic order were developed by \cite{brunk1966maximum} and \cite{dykstra1982maximum}, and the corresponding asymptotic properties were derived by \cite{praestgaard1996asymptotic}. The \emph{hazard rate} order is defined as $G \geq_{HR} F$ if  $\bar{F} / \bar{G}$ is non-increasing, which is equivalent to $f / \bar{F} \geq g / \bar{G}$ in the case of absolutely continuous distributions, where $f$ and $g$ are the densities corresponding to $F$ and $G$. \cite{dykstra1991statistical} studied estimation and inference under a hazard rate order.  Finally, the \emph{likelihood ratio} order is defined as $G \geq_{LR} F$ if $g/f$ is non-decreasing. \cite{dykstra1995inference, yu2017density, mosching2020estimation} and \citet{westling2019nonparametric} considered estimation and inference under a likelihood ratio order. These three canonical examples of stochastic orders are themselves ordered: $G \geq_{LR} F$ implies $G \geq_{HR} F$ implies $G \geq_{ST} F$.

It is natural to ask where the monotone hazard ratio order fits into the hierarchy of the three common stochastic orders. It turns out that the monotone hazard ratio order does not generally imply, nor is it implied by, any of the three common stochastic orders. To show this, we provide continuous and discrete counterexamples for each case. These examples are illustrated in Figure~\ref{Fig: orders_relationships}. We note that the fact that our order is not implied by nor implies these other orders means in particular that previously established properties of and methods for inference under these orders do not apply to the monotone hazard ratio order.

We first show that the monotone hazard ratio order does not imply the usual stochastic order, which further implies that the monotone hazard ratio order does not imply the hazard rate or monotone likelihood ratio orders. Suppose that $S$ and $T$ have Weibull distributions with shape parameters $k_S$ and $k_T$ and scale parameters $\sigma_S$ and $\sigma_T$, respectively. Then the hazard ratio function $\theta (t) = \lambda_S(t) / \lambda_T(t) $ is proportional to $t^{k_S - k_T}$ for $t > 0$, so that $S \geq_{MHR} T$ if and only if $k_S \geq k_T$, and $S \leq_{MHR} T$ if and only if $k_S \leq k_T$. On the other hand, $F_S(t) \leq F_T(t)$ if and only if $t^{k_S - k_T} \leq \sigma_S^{k_S} / \sigma_T^{k_T}$. If $k_S \neq k_T$, then $t \mapsto t^{k_S - k_T}$ ranges from $0$ to $\infty$, which implies that it cannot be the case that either $S \geq_{ST} T$ or $T \geq_{ST} S$. Therefore, if $k_S > k_T$, then $S \geq_{MHR} T$, but $S \ngeq_{ST} T$, which also implies that $S \ngeq_{HR} T$ and $S \ngeq_{LR} T$. Hence, the monotone hazard ratio order does not imply any of these other three common orders in the continuous case (first column of Figure~\ref{Fig: orders_relationships}).
For a counterexample in the discrete case, suppose that $F_S$ follows a geometric distribution with success probability $p_S$ on $\{1,2,\dots\}$, so that $\lambda_S(k) = p_S$ for all $k \in \{1,2,\dots\}$. Hence, $S \geq_{MHR} T$ for any $T$ supported on $\{1,2,\dots\}$ such that $\lambda_T(k)$ is non-increasing in $k$. The usual stochastic order fails to hold if $\bar{F}_S(k) = (1-p_S)^k < \bar{F}_T(k) = \prod_{j=1}^k [1 - \lambda_{T}(j)]$ for any $k \in \{1,2,\dots\}$. Both of these are the case, for instance, if $T$ also follows a geometric distribution with success probability $p_T < p_S$ (second column of Figure~\ref{Fig: orders_relationships}).

We now show that the likelihood ratio order does not imply the monotone hazard ratio order, which further implies that the hazard rate order and usual stochastic order do not imply the monotone hazard ratio order. For an example in the continuous case, suppose that $S$ and $T$ follow Beta distributions with parameters $(\alpha, \beta_S)$ and $(\alpha, \beta_T)$ for $\beta_S < \beta_T$. Then the density ratio is proportional to $(1-t)^{\beta_S - \beta_T}$, which is strictly increasing, so $S \geq_{LR} T$. Furthermore, if $\alpha \in (0,1)$, then one can also show that the hazard ratio function is strictly decreasing, so that $S <_{MHR} T$. Therefore, the likelihood ratio order does not imply the monotone hazard ratio order in the continuous case, so neither do the hazard rate or usual stochastic orders (third column of Figure~\ref{Fig: orders_relationships}). For a counterexample in the discrete case, suppose $S$ has a uniform distribution on $\{t_1 < \cdots < t_K\}$ for $K > 1$ and $T$ satisfies (1) $f_T(t_j) \geq f_T(t_{j+1})$ for $j = 1, \dotsc, K-1$, and (2) $f_T(t_j) > (K-j+1) f_T(t_{j-1}) [1 - f_T(t_{j-1})] / (K-j)$ for all $j = 2, \dots, K-1$. Both (1) and (2) can be achieved simultaneously if and only if $f_T(t_1) \geq 1/K$. Then the ratio of the mass functions is proportional to $f_T$, so the likelihood ratio order holds by assumption (1). However, we can also show that $\lambda_S(t_{j-1}) / \lambda_T(t_{j-1}) > \lambda_S(t_{j}) / \lambda_T(t_{j})$ for all $j = 2, \dots, K - 1$. So the monotone hazard ratio order cannot hold (last column of Figure~\ref{Fig: orders_relationships}). 

\begin{figure}[t!]
    \centering
    \includegraphics[width = 1\textwidth]{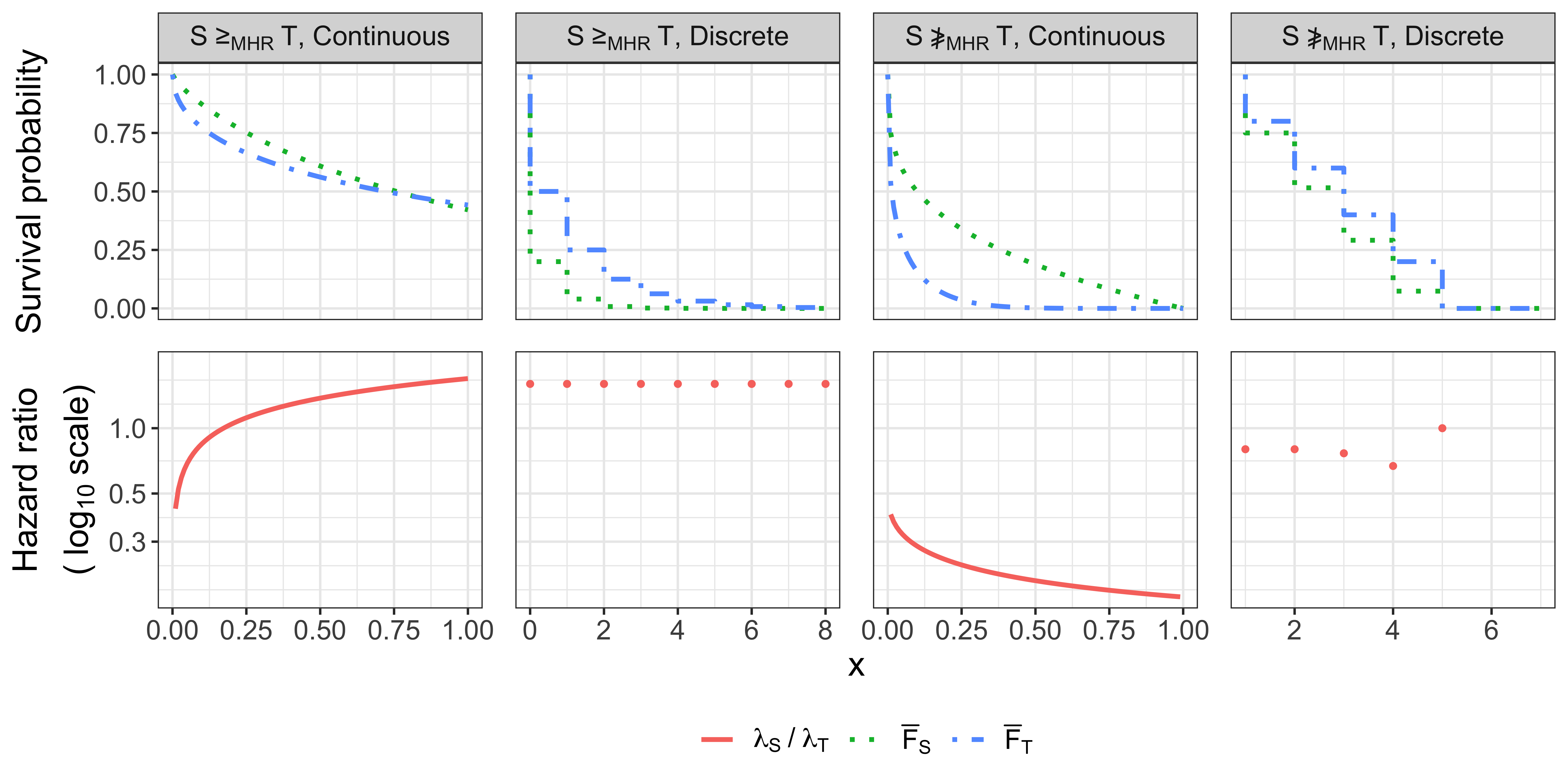}
    \caption{The relationship of $\geq_{MHR}$ to other stochastic orders. The upper row shows the survival functions $\bar{F}_S$ and $\bar{F}_T$ while the lower row shows the hazard ratio function $\lambda_S / \lambda_T$. Column 1: $S \sim \n{Weibull}(0.8, 1.2)$, $T \sim \n{Weibull}(0.5, 1.5)$. Column 2: $S \sim \n{Geometric}(0.8)$, $T \sim \n{Geometric}(0.5)$. Column 3: $S \sim \n{Beta}(0.3, 1)$, $T \sim \n{Beta}(0.3, 6)$. Column 4: $S \sim \n{Uniform}\{1,2,\dotsc,5\}$, $T \sim$ non-increasing discrete distribution defined in the text.}
    \label{Fig: orders_relationships}
\end{figure}

One special case where the monotone hazard ratio order does imply the hazard rate order, and therefore the usual order as well, is when $\lim_{t \to t_{\n{max}}} \lambda_S(t) / \lambda_T(t) \leq 1$, where $t_{\n{max}} := \sup\{ \n{Supp}(F_S) \cup \n{Supp}(F_T)\}$. This is the case, for instance,  when a treatment is known to be non-toxic, or when a harmful exposure is known to never be beneficial.  In particular, if $F_S$ and $F_T$ are supported on the same finite discrete set $\{t_1 < t_2 <\cdots < t_K\}$ and $f_S(t_K) > 0$ and $f_T(t_K) > 0$, then necessarily $\lambda_S(t_K) = \lambda_T(t_K) = 1$ so $S \geq_{MHR} T$ implies $S \geq_{HR} T$.

\section{Nonparametric inference with right-censored data}\label{sec:estimation}

\subsection{Statistical setting}

In this section, we provide an estimator of a monotone hazard ratio function $\theta$ using independently right-censored data. We derive the asymptotic distribution of our estimator, and use this result to construct asymptotically valid pointwise confidence intervals for $\theta$. 

For each $i \in \{1,\dotsc,n\}$, we let $A_i \sim \n{Bernoulli}(\pi)$ indicate the cohort for unit $i$. For a randomized study, $A_i = 0$ corresponds to control, and $A_i = 1$ corresponds to treatment, though the data need not be from a randomized trial. We assume that $\pi \in (0,1)$. For $i$ such that $A_i = 1$, we let $S_i \sim F_S$ be the event time and $U_i \sim F_U$ be the censoring time. For $i$ such that $A_i = 0$, we let $T_i \sim F_T$ be the event time and $V_i \sim F_V$ be the censoring time. We assume that $S_i$ and $U_i$ are independent and $T_i$ and $V_i$ are independent for each $i$ --- that is, the censoring is independent of the event within each treatment arm.  If $A_i = 1$, we  observe the right-censored data $Y_i := \min\{S_i, U_i\}$ and $\Delta_i := I(S_i \leq U_i)$, and if $A_i = 0$, we observe $Y_i := \min\{T_i, V_i\}$ and $\Delta_i := I(T_i \leq V_i)$. The observed data for unit $i$ is then $O_i := (Y_i, \Delta_i, A_i)$, and we assume that $O_1, \dotsc, O_n$ are IID. 
 
When $F_S$ and $F_T$ are discrete, the hazard ratio function can be estimated using the ratio of the empirical hazard functions within each treatment arm. The empirical hazard functions converge at the rate $n^{-1/2}$ to normal limits, so by the delta method, their ratio does as well. Hence, inference for the hazard ratio function in this case can be obtained using standard methods. Furthermore, monotonicity of the hazard ratio function can be enforced by projecting the empirical estimator onto the space of monotone functions \citep{westling2020correcting}. Therefore, here, we focus on the more challenging case where $F_S$ and $F_T$ are absolutely continuous distributions. We make no assumptions about the censoring distributions $F_U$ and $F_V$.

\subsection{Proposed estimator}

Our estimator is based on the representation of $\theta$ presented in Theorem~\ref{thm:mhr_characterization}. We recall from Theorem~\ref{thm:mhr_characterization} that if $F_S \ll F_T$ and $\theta$ is non-decreasing and continuous on the support of $F_T$, then we can represent $\theta$ in terms of the cumulative hazard functions $\Lambda_S$ and $\Lambda_T$ as $\theta = \partial_- \n{GCM}_I(\Lambda_S \circ \Lambda_T^-) \circ \Lambda_T$, where $I$ is the smallest closed interval containing $\n{Im}(\Lambda_T)$. Our estimator is defined by replacing the unknown elements in this representation with nonparametric estimators thereof. We let $\Lambda_{S,n}$ be the stratified Nelson-Aalen estimator \citep{nelson1969hazard, aalen1978nonparametric} of the cumulative hazard function $\Lambda_S$ based on the cohort for which $A = 1$. Similarly, we let $\Lambda_{T,n}$ be the stratified Nelson-Aalen estimator of $\Lambda_T$ based on the control cohort for which $A = 0$. We also define $\eta_n := \Lambda_{T,n}(\gamma_n)$, where $\gamma_n$ is the minimum of the empirical $1-r_n$ quantile of the $Y_i$'s for which $A_i = 0$ and the empirical $1-r_n$ quantile of the $Y_i$'s for which $A_i = 1$, where $r_n > 0$ is a non-increasing sequence converging to $r \geq 0$. It follows that $\gamma_n$ is converging to $\gamma$, the minimum of the $(1-r)$th quantile of $Y$ given $A = 1$ and the $(1-r)$th quantile of $Y$ given $A = 0$. Additional conditions on $r_n$ and practical suggestions for setting $r_n$ will be provided below. We then define our estimator $\theta_n$ of $\theta$ as
\[\theta_n := \partial_-\n{GCM}_{[0,\eta_n]}(\Lambda_{S, n} \circ \Lambda_{T, n}^-) \circ \Lambda_{T, n}\ . \]
It is straightforward to compute $\theta_n$ using standard software packages. Specifically, in the statistical computing software \texttt{R} \citep{r-core}, the Nelson-Aalen estimators $\Lambda_{S, n}, \Lambda_{T, n}$ can be obtained using the package \texttt{survival} \citep{survival-package}, and the slopes of the greatest convex minorant of $\Lambda_{S, n} \circ \Lambda_{T, n}^-$ can be obtained using the package \texttt{fdrtool} \citep{fdrtool-package}. Code for computing $\theta_n$ is provided in Supplementary Material.

\cite{kim2011bayesian} proposed a nonparametric Bayesian approach to estimating a monotone hazard ratio function. Their model permits either monotone non-decreasing or non-increasing hazard ratio functions, whereas the type of monotonicity must be known a priori for our estimator. Their model also allows for the incorporation of covariates, which we have have not explored. However, approximating the posterior distribution in their model is complicated and possibly computationally intensive, in contrast to the simple implementation of our procedure.

\subsection{Convergence in distribution}

We now demonstrate that $n^{1/3} \left[ \theta_n(x) - \theta_0(x)\right]$ converges in distribution for fixed $x$ to a scaled Chernoff distribution. The (standard) Chernoff distribution is defined as the derivative at zero of the GCM of a Brownian motion plus a quadratic; i.e. $W := [\partial_-\n{GCM}_{\d{R}}(Z)](0)$, where $Z(t) := B(t) + t^2$ for $B$ a standard two-sided Brownian motion with $B(0) = 0$. 

\begin{thm}\label{thm: convergence in distribution}
Suppose $x \in (0,\gamma)$ is such that that $F_S$, $F_T$, and $\theta$ are continuously differentiable at $x$ with finite and strictly positive derivatives, and $F_S$, $F_U$, $F_T$, and $F_V$ are $< 1$ in a neighborhood of $x$. Also suppose that there exist $\varepsilon, C > 0$ such that $r_n \geq C (\log n)^{2+\varepsilon} / n$ for all $n$. Then
\[
    n^{1/3} \left[\theta_n(x) - \theta(x)\right] \indist \left\{\frac{4\theta'(x)\kappa(x)}{\lambda_T(x)^2}\right\}^{1/3}W\ ,
\]
where $W$ follows the Chernoff distribution and \[\kappa(x) := \theta(x)\left[\frac{\lambda_T(x)}{\pi\bar{F}_S(x)\bar{F}_U(x)} + \frac{\lambda_S(x)}{(1-\pi)\bar{F}_T(x)\bar{F}_V(x)}\right]\ .\]
\end{thm}

Due to its connection with GCMs, the Chernoff distribution appears in the asymptotic distribution of summaries of many monotonicity-constrained estimators (e.g., \citealp{groeneboom1985estimating, huang1995estimation, westling2019nonparametric}, among many others). The properties of the Chernoff distribution were studied extensively by \cite{groeneboom2001computing}. In particular, common quantiles of the distribution are tabulated therein, which facilitates the construction of asymptotic confidence intervals for $\theta$ using Theorem~\ref{thm: convergence in distribution}, as we discuss below.

Theorem~\ref{thm: convergence in distribution} implies that $\theta_n(x)$ converges to $\theta(x)$ at the rate $n^{-1/3}$. This is slower than the rate $n^{-2/5}$ achieved by estimators of the hazard function based on kernel smoothing with optimal bandwidth selection \citep{muller1990locally, Groeneboom2010smoothed}. However, this latter result requires that the hazards possess two continuous derivatives, while Theorem~\ref{thm: convergence in distribution} only requires one continuous derivative of the hazard ratio. In addition, asymptotically valid inference using estimators based on kernel smoothing is challenging due to bias arising in the limit distribution \citep{calonico2018effect}.

Theorem~\ref{thm: convergence in distribution} requires that $r_n$ not converge too quickly to zero, meaning that the upper limit of the region over which the GCM is taken not converge too quickly to the upper limit of support of the observed times. This ensures that $\Lambda_{T,n}$ and $\Lambda_{S,n}$ are uniformly consistent on the increasing interval $[0,\gamma_n]$ \citep{stute1994strong}. The requirement is satisfied if, for instance, $r_n = r >0$ for all $n$, or if $r_n = (\log n)^{2+\varepsilon} / n$ for some $\varepsilon > 0$. In practice, we recommend setting $r_n = 0.05$ for $n < 1000$, and $r_n = (\log n)^{2.1} / n$ for $n \geq 1000$.

\cite{kim2011bayesian} proposed a nonparametric Bayesian approach to estimating a monotone hazard ratio function. Their model permits either monotone non-decreasing or non-increasing hazard ratio functions. The type of monotonicity must be known a priori for our estimator, but we expect that in most cases where monotonicity can be assumed, the direction of monotonicity is also known. Approximating the posterior distribution in their model is complicated and possibly computationally intensive, in contrast to the simple implementation of our procedure. \cite{kim2017analysis} proved that the rate of convergence of the posterior distribution of the nonparametric Bayesian estimator proposed by \cite{kim2011bayesian} is $(n / \log n)^{-1/3}$, which is just a poly-log factor slower than the rate of convergence of our estimator. However, to the best of our knowledge, it is not known whether the posterior distribution of the estimator proposed by \cite{kim2011bayesian} yields asymptotically calibrated confidence intervals for $\theta(x)$. In the next section, we use Theorem~\ref{thm: convergence in distribution} to construct asymptotically valid pointwise intervals using our estimator.

\subsection{Construction of confidence intervals} \label{sec:inference}

We propose two methods of constructing confidence intervals for $\theta$. The first method is based on the asymptotic distribution of $\theta$ provided in Theorem~\ref{thm: convergence in distribution}. By Theorem~\ref{thm: convergence in distribution}, a Wald-type asymptotic $(1-\alpha)$-level confidence interval for $\theta$ is given by $\theta_n \pm \tau_n(x) q_{1-\alpha/2} / n^{1/3}$, where $\tau_n(x)$ is a consistent estimator of $\tau(x) := \{ 4\theta'(x)\kappa(x) / \lambda_T(x)^2\}^{1/3}$, and $q_{p}$ is the $p$th quantile of the standard Chernoff distribution. Quantiles of the Chernoff distribution are tabulated in \cite{groeneboom2001computing}. We note that $\tau(x)$ involves both $\lambda_S(x)$ and $\lambda_T(x)$, so one approach to estimating $\tau(x)$ would be to plug in consistent estimators of $\lambda_S(x)$ and $\lambda_T(x)$. Instead, we rewrite $\tau(x)$ as
\begin{align*}
    \tau(x) &= \left\{4 (\theta \circ \Lambda_T^-)' \circ \Lambda_T(x) \left[\frac{\theta(x)}{\pi\bar{F}_S(x)\bar{F}_U(x-)} + \frac{\theta(x)^2}{(1-\pi)\bar{F}_T(x)\bar{F}_V(x-)}\right]\right\}^{1/3}.
\end{align*}
This form of $\tau(x)$ no longer depends directly on $\lambda_S$ or $\lambda_T$. In this expression, $\theta_n$, $\Lambda_{T,n}$, and the Kaplan-Meier estimators $F_{S,n}$, $F_{U,n}$, $F_{T,n}$, and $F_{V,n}$ can be substituted for their true counterparts in constructing an estimator $\tau_n(x)$ of $\tau(x)$. Hence, the only remaining challenge is to estimate $(\theta \circ \Lambda_T^-)'$. We do this using the derivative estimator obtained by applying a local linear kernel smoother to the set of points $\{( u_k, \theta_n \circ \Lambda_{T, n}^- (u_k)) : k =1, \dotsc, m_n\}$, where $m_n = \lceil n^{2/3} \rceil$, and $\{0 = u_1 < u_2 < \cdots < u_{m_n} = \eta_n\}$ is a uniform grid on $[0, \eta_n]$. We choose the bandwidth for the kernel smoother using cross validation \citep{Guidoum2012kedd}.

Sample splitting has also been shown to yield valid inference and reduced variance for estimators with $n^{-1/3}$-rate asymptotics without the need to estimate additional nuisance parameters in the limit distribution \citep{banerjee2005confidence, banerjee2019divide}. To implement this method, the $n$ observations are first split randomly into $m$ disjoint subsets of approximately equal size. The estimator $\theta_{n,j}$ is then computed for each subset $j \in \{1, \dotsc, m\}$. These estimators are averaged to obtain a pooled estimator $\bar{\theta}_{n,m} = \frac{1}{m}\sum_{j = 1}^{m} \theta_{n,j}$. Finally, an asymptotic $(1-\alpha)$-level confidence interval for $\theta(x)$ is given by $\bar{\theta}_{n,m}(x) \pm t_{1-\alpha/2, m-1} \sigma_{n,m}(x) / \sqrt{m}$, where $\sigma_{n,m}(x)$ is the empirical standard deviation of the $m$ subset estimators $\{\theta_{n,1}(x), \dotsc, \theta_{n,m}(x)\}$ and $t_{p,k}$ is the $p$th quantile of the $t$ distribution with $k$ degrees of freedom.

\section{Numerical studies\label{sec:simulation}}

To assess the finite-sample performance of our proposed estimator and confidence intervals, we performed the following numerical study. We simulated data from three different scenarios corresponding to linear, convex, and concave $\theta$. Defining $\lambda(x) := 0.25+ \sin^2(6\pi x)$, in the linear case, we set  $\lambda_S(x) = x\lambda(x)$ and $\lambda_T(x) = \lambda(x)$, so that $\theta(x) = x$.  In the convex case, we set $\lambda_S(x) = x^2\lambda(x)$ and $\lambda_T(x) = \lambda(x)$, so that $\theta(x) = x^2$. In the concave case, we set $\lambda_S(x) = x\lambda(x)$ and $\lambda_T(x) = \sqrt{x}\lambda(x)$, so that $\theta(x) = \sqrt{x}$. Notably, $\lambda_S(x) > 0$ and $\lambda_T(x) > 0$ for all $x > 0$, and are multiples of a periodic function due to the inclusion of $\sin^2$. This is common in many applications where event rates follow weekly, monthly, or seasonal trends. For the censoring distributions, we set both $F_U(t)$ and $F_V(t)$ as  $1- e^{-0.1t}$ for $0 \leq t < 1$, $1 - e^{-0.15t}$ for $1 \leq t < 2$, and $1$ for $t \geq 2$. Hence, the censoring distributions are mixed discrete-continuous distributions supported on $[0,2]$, and have discrete components at 1 and 2 with probabilities 0.044 and 0.078, respectively. Finally, we set $\pi = 0.5$.


\begin{figure}
    \centering
    \includegraphics[width = .95\linewidth]{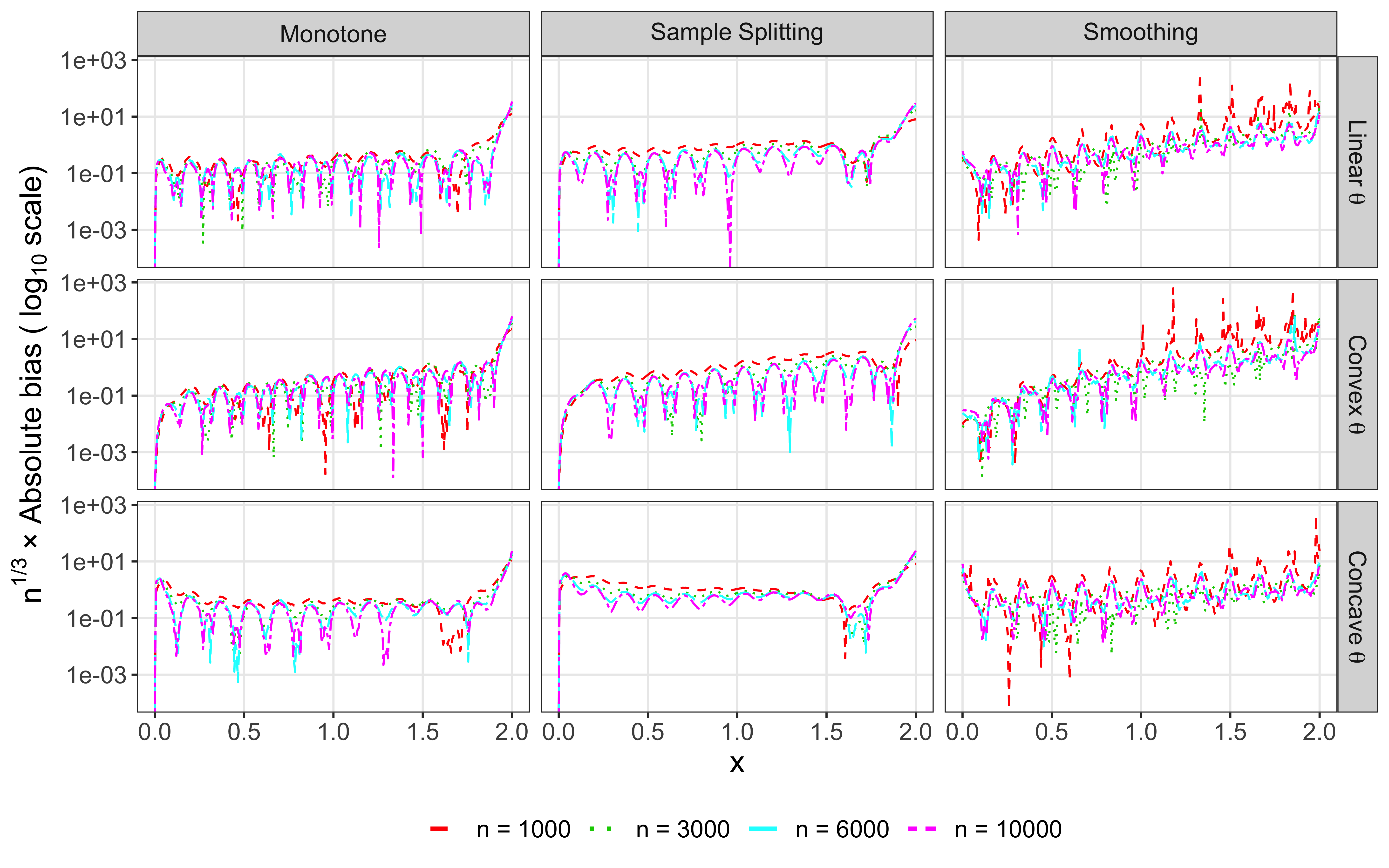}
     \includegraphics[width = .95\linewidth]{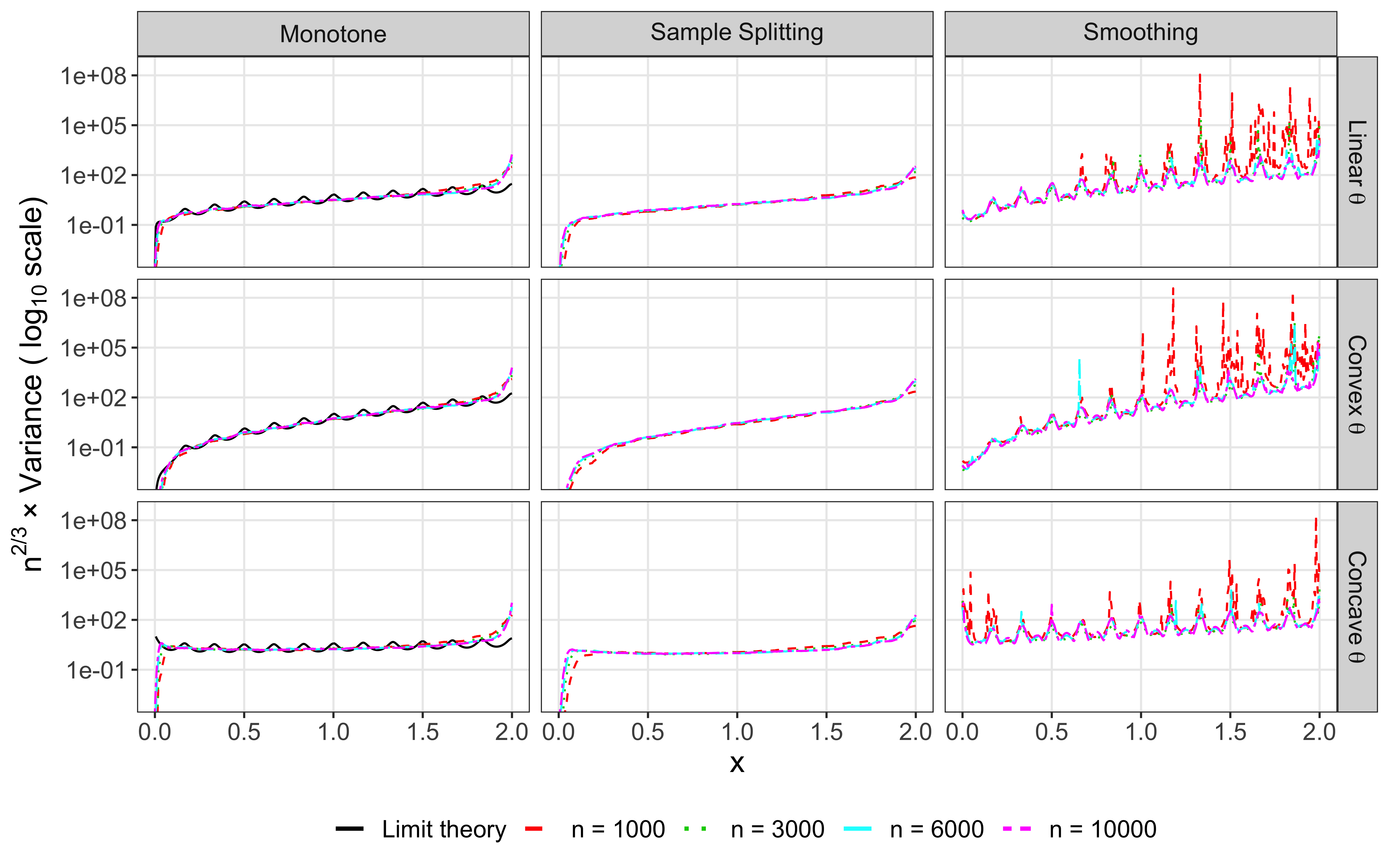}
    \caption{Top: absolute bias of the three estimators scaled by $n^{1/3}$ as a function of $x$. Bottom: variance of the three estimators scaled by $n^{2/3}$ as a function of $x$. The rows correspond to linear, convex and concave $\theta$. The first column is our estimator, the second column is the sample splitting estimator, and the third column is the kernel smoothing estimator.}
    \label{Fig: bias}
\end{figure}

For each sample size $n$ equal to 1000, 3000, 6000, and 10000, we simulated 1000 right-censored datasets for each of the three mechanisms described above. For each dataset and for each $x$ equal to $0.005, 0.01, \dotsc, 2$, we computed our proposed estimator $\theta_n(x)$, the sample splitting estimator $\bar{\theta}_{n,m}(x)$ with $m = 5$ splits, and the corresponding confidence intervals defined in Section~\ref{sec:inference}. For comparison, we also computed an estimator and confidence intervals based on taking the ratio of kernel smoothing estimators of the individual hazard functions, which does not require or enforce monotonicity of $\theta$ \citep{watson1964hazard}. For the kernel smoothing estimators of the hazard functions, we used the Epanechnikov kernel and selected the bandwidths using cross validation. We did not compare our procedure to that of \cite{kim2011bayesian} due to the lack of availability of computer code implementing their procedure.

The top panel of Figure~\ref{Fig: bias} displays $n^{1/3}$ times the absolute bias of the three estimators as a function of $x$. Figure~\ref{Fig: relative bias and variance} in Supplementary Material displays the relative absolute bias of the smoothing and sample splitting estimators to our estimator. The scaled bias of all three estimators generally decreases with sample size, which aligns with the expectation that the biases decrease faster than $n^{-1/3}$ for $x \in (0,2)$. The absolute bias of the three estimators exhibits periodicity inherited from the periodicity of the underlying hazard functions. All three estimators exhibit large bias near $x = 2$, which is expected given the challenges of estimation near the boundary of support. The bias near $x = 0$ is highest for all three estimators in the concave $\theta$ case, and the bias for $x$ between 1 and 2 is largest in the convex case, which makes sense because this is when the derivative of $\theta$ is large.

For most values of $x$, our estimator has slightly smaller absolute bias than the sample splitting estimator, especially for $n = 1000$, which is expected because the sample splitting estimator inherits the bias of our estimator with one-fifth the sample size. The absolute bias of the smoothing-biased estimator relative to that of our estimator is generally proportional to the magnitude of the second derivatives of $\lambda_S$ and $\lambda_T$. The absolute bias of our estimator also generally improves relative to that of the smoothing-based estimator as $x$ increases. We believe this is due to a combination of the monotonicity assumption and censoring. As $x$ increases, the effective sample size decreases as a result of right-censoring, which generally increases bias. However, the monotonicity assumption of our estimator may aid in reducing this bias by using information from earlier time-points, unlike the smoothing-based estimator.

The bottom panel of Figure~\ref{Fig: bias} displays $n^{2/3}$ times the absolute bias of the three estimators as a function of $x$.  Figure~\ref{Fig: relative bias and variance} in Supplementary Material displays the relative variance of the smoothing and sample splitting estimators to our estimator. The variance of our estimator is close to the theoretical limit except for $x$ near 2 for all values of $n$.  The empirical variance does not capture the periodic pattern of the true variance, but we expect it would at larger sample sizes. The variance of the sample splitting estimator is a constant factor smaller than the variance of our estimator, as expected based on the theory of \cite{banerjee2019divide}. The variance increases as a function of $x$ fastest for the convex case, followed by the linear and concave cases. This is due to the appearance of $\theta(x)$ and $\theta'(x)$ in the scale parameter in the limit distribution established in Theorem~\ref{thm: convergence in distribution}. Both of these values are increasing fastest for the convex case. The variance of the smoothing-based estimator is greater than the variance of our estimator in the sample sizes we considered. However, the relative variance of the smoothing-based estimator improves with sample size because the variance of the smoothing-based estimator goes to zero faster than the variance of our estimator. Overall, the mean squared error of our estimator is no worse than that of the smoothing based estimator for all values of $x$ and sample sizes we considered (Figure~\ref{Fig: relative mse} in Supplementary Material).

\begin{figure}[t]
    \centering
    \includegraphics[width = 1\textwidth]{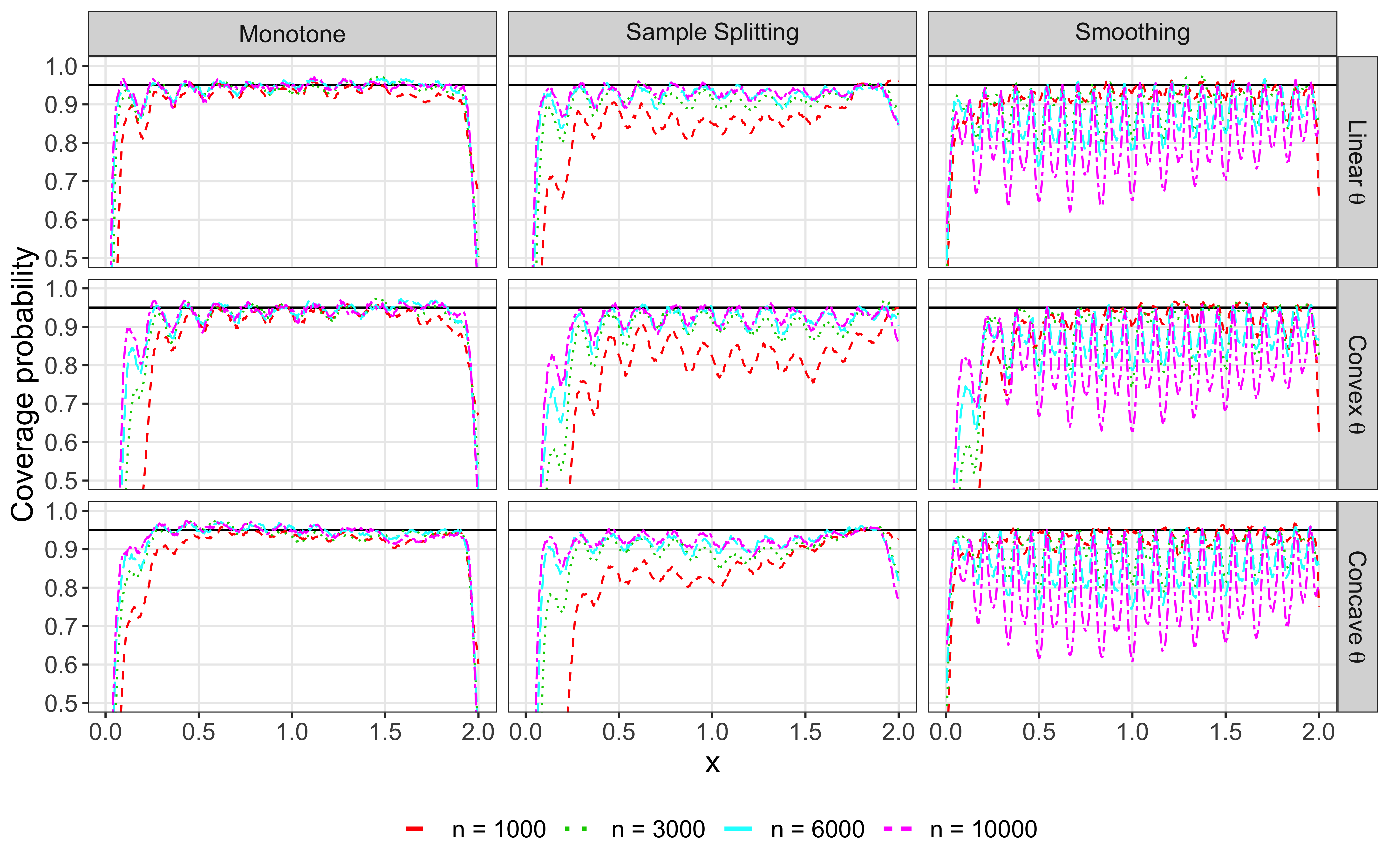}
    \caption{Empirical coverage probabilities of nominal 95\% confidence intervals for the three estimators as a function of $x$. The rows correspond to linear, convex and concave $\theta$. The first column is our estimator, the second column is the sample splitting estimator, and the third column is the kernel smoothing estimator.}
    \label{Fig: coverage}
\end{figure}

Figure~\ref{Fig: coverage} shows the coverage probability of nominal 95\% confidence intervals for the three estimators. 
The coverage of the plug-in intervals centered around our estimator have close to nominal coverage for values of $x$ not too close to 0 or 2. For values of $x$ close to 0, the coverage of the plug-in method is poor due to the difficulty of estimating the derivative $\theta'$ in this region. The sample splitting method has poor coverage for $n=1000$ due to high bias, but the coverage converges to the nominal level as the sample size increases. The smoothing-based estimator has poor coverage for values of $x$ where the second derivatives of the hazard functions are large, which is due to the bias of the smoothing-based estimator.


\section{Analysis of treatment of pulmonary adenocarcinoma\label{sec:application}}

In this section, we use the methods developed in this article to estimate the all-cause mortality hazard ratio of two treatments for pulmonary adenocarcinoma: gefitinib and carboplatin–paclitaxel. Carboplatin-paclitaxel is a type of intravenous chemotherapy, usually taken over a three-hour period once every three weeks for approximately six cycles \citep{herbst2004gefitinib}. Like many chemotherapies, carboplatin-paclitaxel is an invasive treatment that can have severe adverse side effects. Gefitinib is a kinase inhibitor that is taken orally as a tablet once per day. Gefitinib is hence less invasive than carboplatin-paclitaxel, but can also cause adverse side effects. We refer the reader to \cite{mok2009gefitinib} and \cite{inoue2013updated} for additional details about these treatments.

We re-analyzed the results of a clinical trial comparing gefitinib and carboplatin–paclitaxel first reported in \cite{mok2009gefitinib}.  The cohort consisted of $n=1217$ adults with stage IIIB or IV non–small-cell lung cancer with histologic features of adenocarcinoma, and who were nonsmokers or former light smokers and had no previous chemotherapy or biologic or immunologic therapy. These patients  were randomly assigned to gefitinib (609 patients) or carboplatin–paclitaxel (608 patients). 
Treatment for both groups continued until progression of the disease, development of unacceptable toxic effects, a request by the patient or physician to discontinue treatment, serious noncompliance with the protocol, or completion of six chemotherapy cycles. The event time of interest was the time from randomization to the earliest sign of disease progression or death from any cause. Additional details of the trial and cohort design can be found in \cite{mok2009gefitinib}. Since the raw data from this trial are unavailable, we used the event and censoring times reconstructed by \cite{argyropoulos2015analysis} from the published Kaplan-Meier estimates.

Starting from the beginning of treatment, the 12-month estimated survival rates were 24.9\% (95\% CI: 21.4, 29.4) with gefitinib and 6.7\% (95\% CI: 4.3, 8.9) with carboplatin–paclitaxel, suggesting that gefitinib was more effective in preventing the progression of pulmonary adenocarcinoma. \cite{mok2009gefitinib} also estimated a Cox proportional hazard model with treatment by gefitinib, smoking history and gender and obtained a hazard ratio of 0.74 (95\% CI: 0.65, 0.85) corresponding to treatment with gefitinib. They concluded that gefitinib was superior to carboplatin–paclitaxel for treating pulmonary adenocarcinoma.

Although their experimental results confirmed that assignment to gefitinib yielded higher overall 12-month survival probability, the survival curves of the two groups crossed, which suggests that the proportional hazards assumption is violated. Hence, it is of interest to estimate the hazard ratio over time to assess the time-varying effect of gefitinib relative to carboplatin–paclitaxel. The left panel of Figure~\ref{Fig: application} displays the Nelson-Aalen estimators of the cumulative hazard function for the gefitinib cohort versus that of the carboplation-paclitaxel cohort, and its GCM. This plot suggests that it is reasonable to believe that the hazard ratio function is monotone. Furthermore, prior estimates of the hazard ratio function have also suggested that it is monotone \citep{argyropoulos2015analysis}. Here, we estimate the hazard ratio using our monotone estimator, and construct confidence intervals using the plug-in method described in Section~\ref{sec:inference}.

\begin{figure}[ht!]
    \centering
    \includegraphics[width = 1\textwidth]{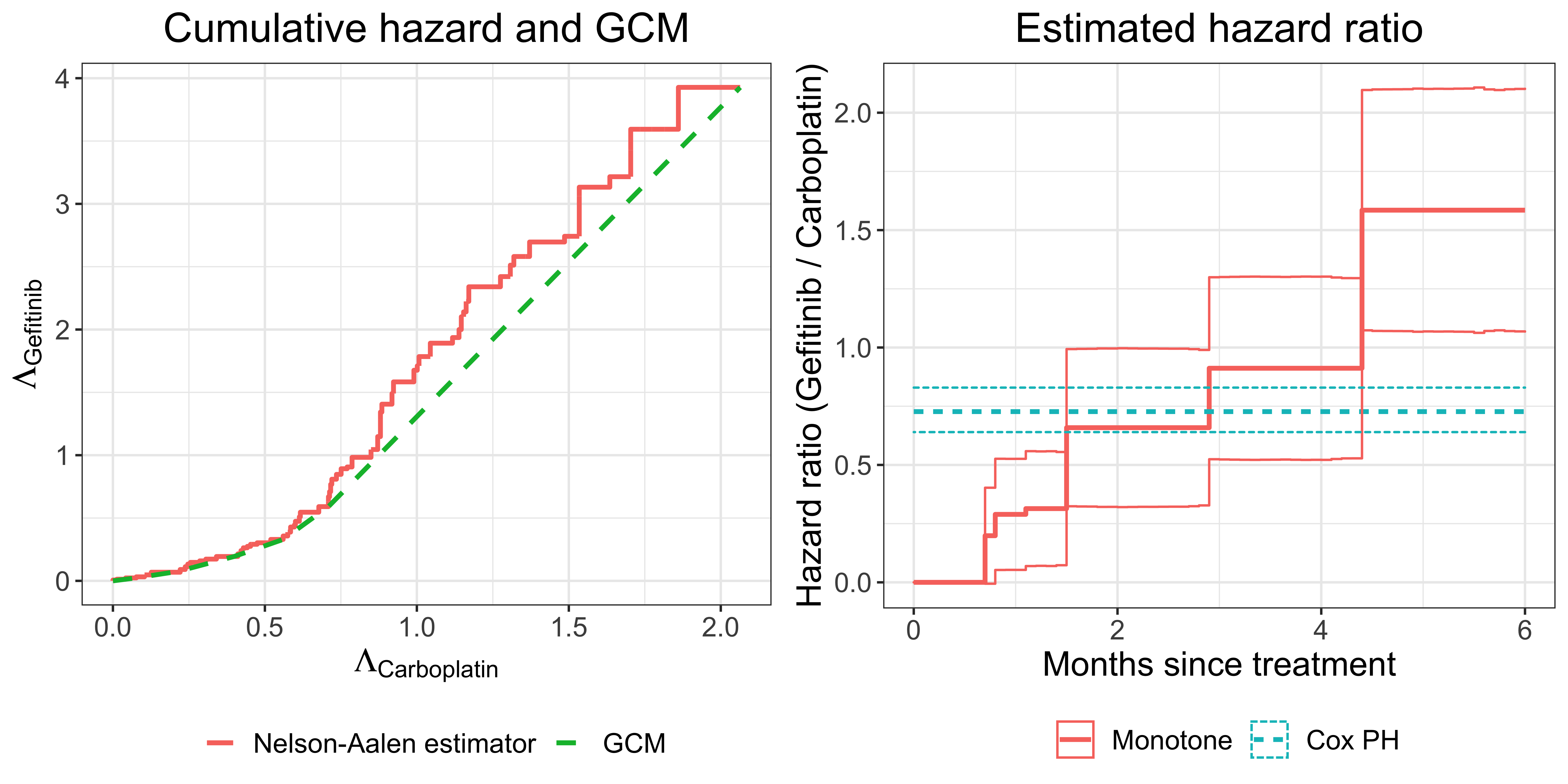}
    \caption{Results of the analysis of the pulmonary adenocarcinoma data. Left panel: the Nelson-Aalen estimator for the gefitinib group plotted against that of the carboplation-paclitaxel group, along with the corresponding greatest convex minorant. Right panel: estimated hazard ratio function and 95\% pointwise confidence intervals using our method and the Cox proportional hazards model.}
    \label{Fig: application}
\end{figure}

The right panel of Figure~\ref{Fig: application} displays the estimated hazard ratio of gefitinib versus carboplation-paclitaxel, as well as the constant hazard ratio estimated by the proportional hazard model. The hazard ratio is only shown through month six, since the estimated curve is flat thereafter. 
We estimate that the hazard ratio increases to one over the span of four months, after which it increases to 1.6 (95\% CI: 1.07, 2.10). Hence, we find evidence that the hazard of disease progression for patients assigned to gefitinib is lower than that of patients assigned to carboplation-paclitaxel through four months post-randomization, but is greater after four months. This could be due to a stronger early benefit of gefitinib and a delayed effect of carboplation-paclitaxel. Alternatively, it could be due to heterogeneous effects of carboplation-paclitaxel relative to gefitinib. For example, frailer patients may have been more likely to progress quickly taking carboplation-paclitaxel than taking gefitinib, leaving a less frail cohort with better survival prospects after four months.

\singlespacing
\section*{Acknowledgements}
The authors  gratefully acknowledge support from the University of Massachusetts Amherst Department of Mathematics and Statistics startup fund (TW) and NSF Award 2113171 (TW). The authors are also grateful for thoughtful feedback from Anna Liu, John Staudenmayer, and Marco Carone.

\singlespacing
\bibliographystyle{apa}
\bibliography{mybib.bib}

\clearpage

\begin{adjustwidth}{-.25in}{-.25in}

\section*{Supplementary Material}\vspace{.1in}

\subsection*{Supplementary figures}

\begin{figure}[htbp!]
    \centering
    \includegraphics[width = 1\textwidth]{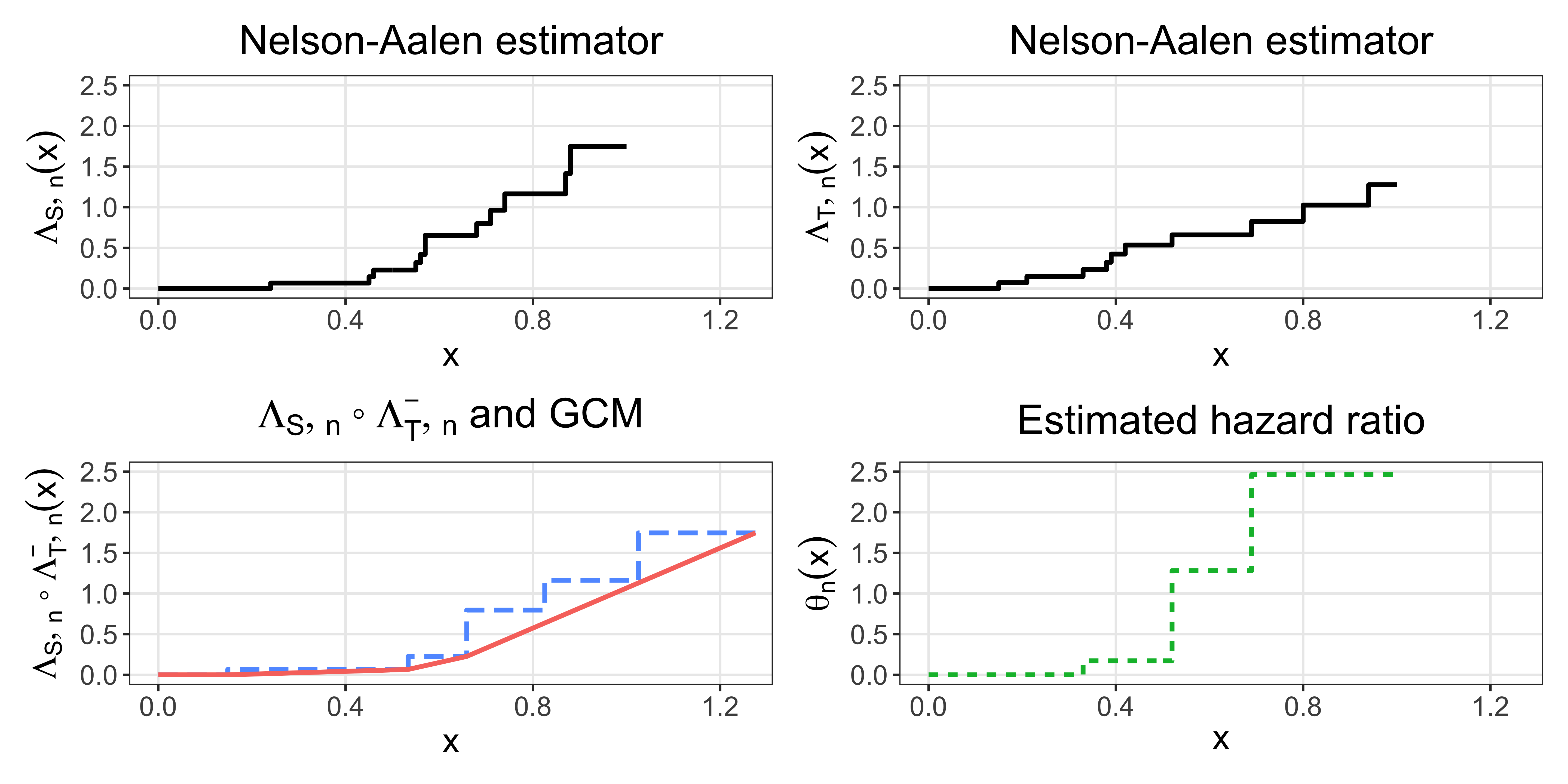}
    \caption{A toy example showing the relationship between Nelson-Aalen estimators, GCM and our hazard ratio estimator $\theta_n$. Left upper panel: the Nelson-Aalen estimor $\Lambda_{S,n}$. Right upper panel: the Nelson-Aalen estimor $\Lambda_{T,n}$. Left lower panel: $\Lambda_{S,n} \circ \Lambda_{T,n}^-$ (solid line) along with its GCM (dashed line). Right lower panel: estimated hazard ratio function $\theta_n$.}
    \label{Fig: toy example}
\end{figure}

\begin{figure}[htbp!]
    \centering
    \includegraphics[width = 0.95\textwidth]{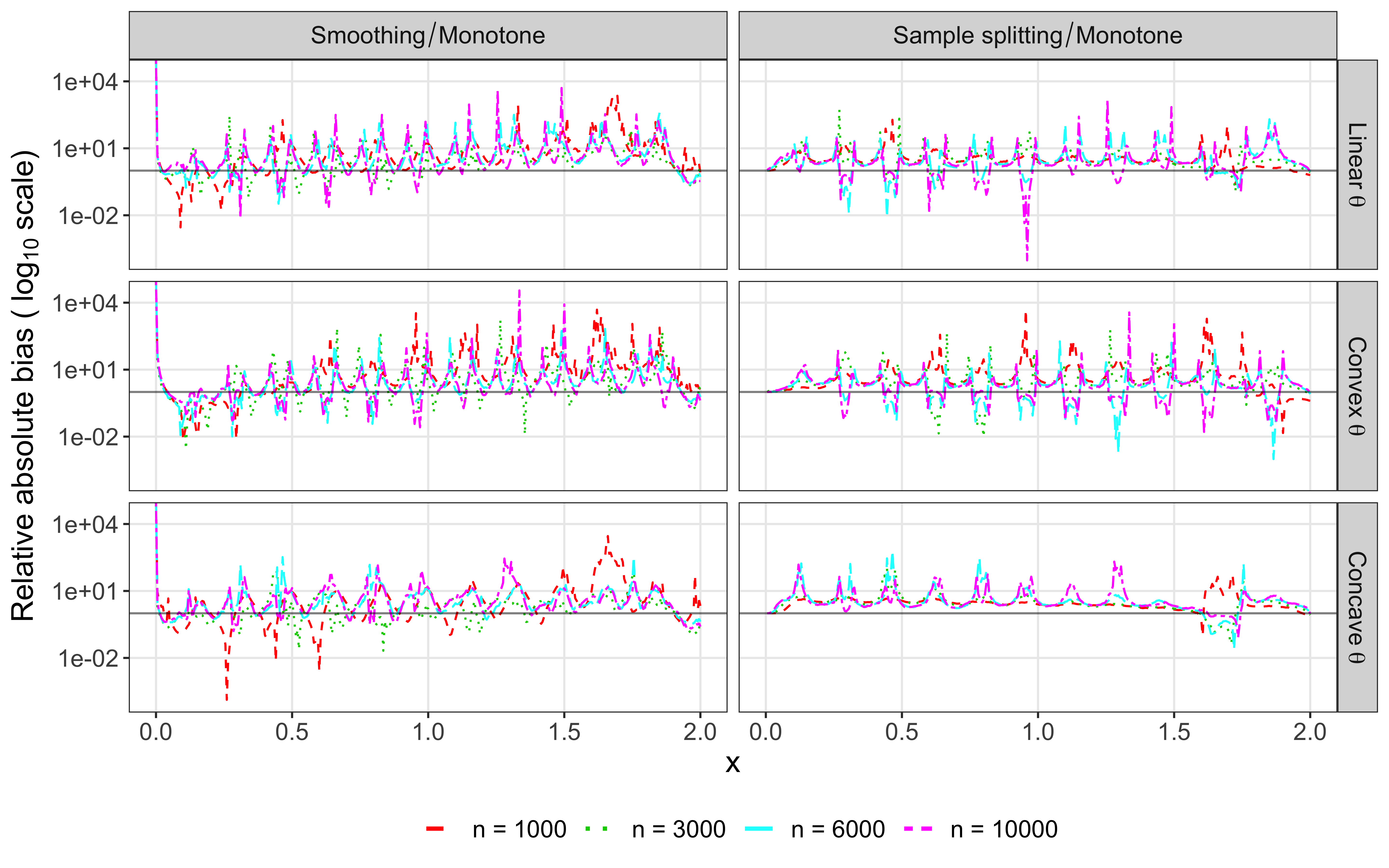}
    \includegraphics[width = 0.95\textwidth]{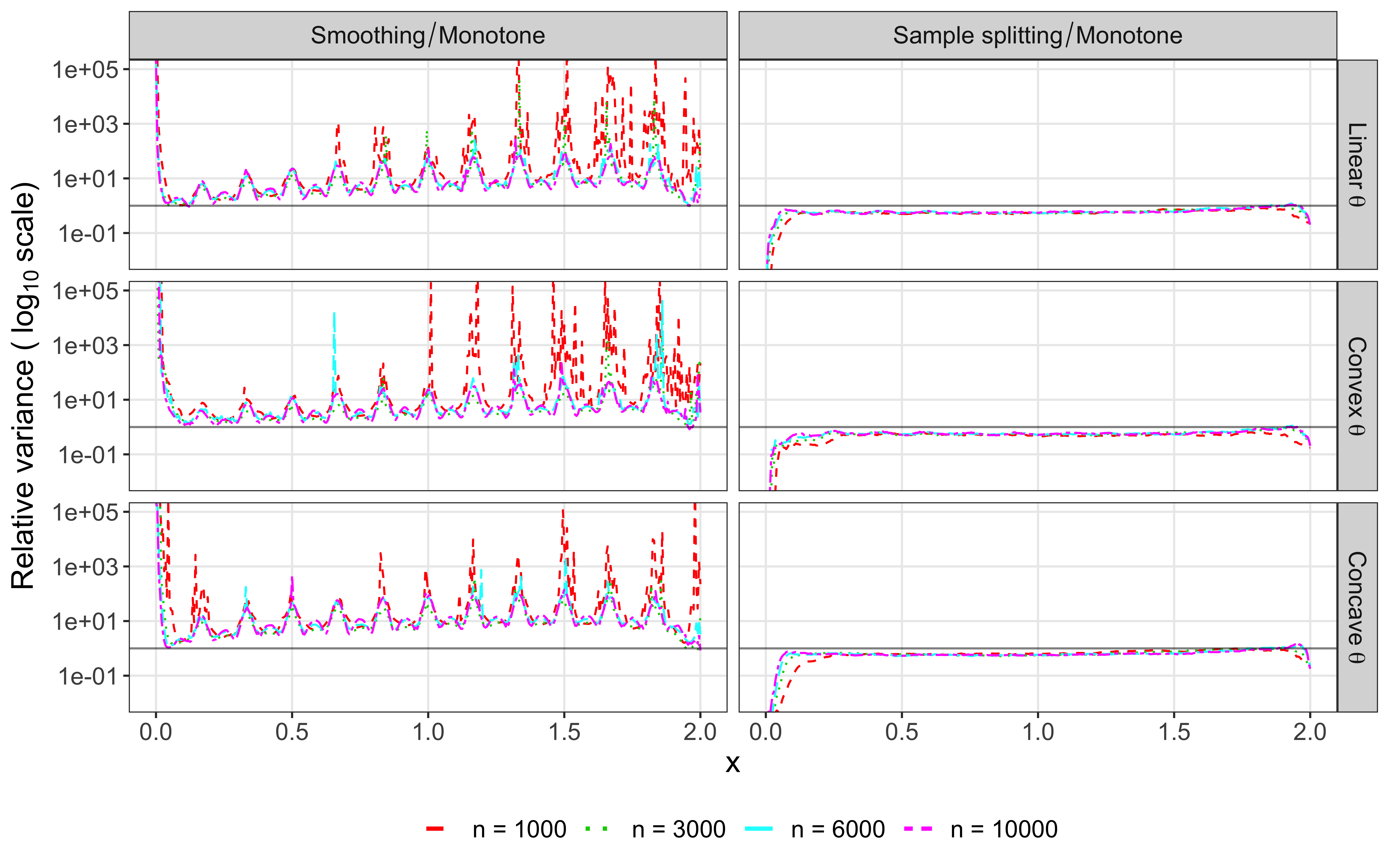}
    \caption{Relative bias and variance of the estimators. Upper panel: The relative absolute bias of our monotone estimator over that of the smoothing estimator (left column) and sample splitting estimator (right column). Lower panel: The relative variance of our monotone estimator over that of the smoothing estimator and sample splitting estimator. The solid gray horizontal line represents a ratio of 1.}
    \label{Fig: relative bias and variance}
\end{figure}

\begin{figure}[htbp!]
    \centering
    \includegraphics[width = 0.95\textwidth]{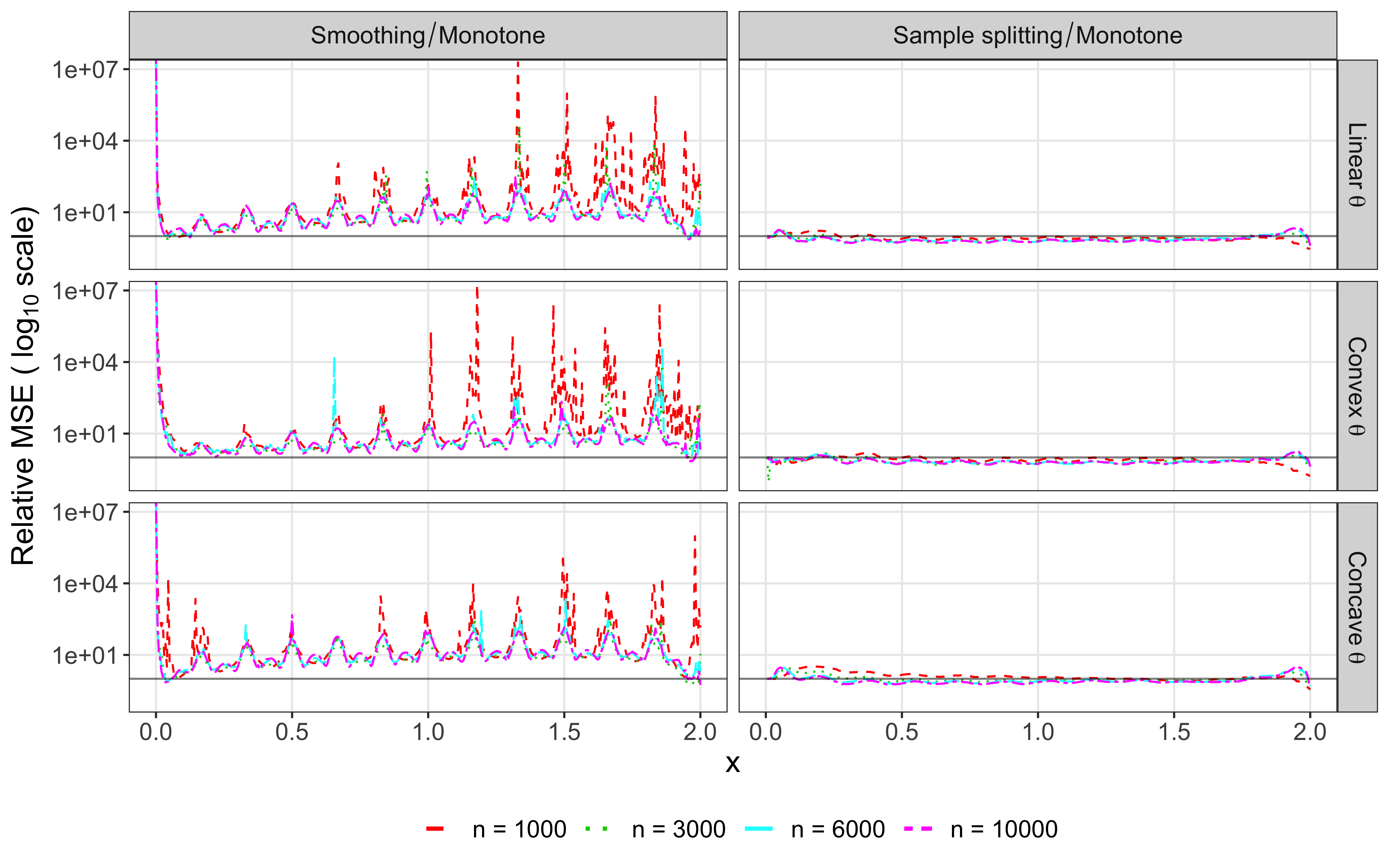}
    \caption{Relative mean squared error (MSE) of the estimators. Left column: The relative MSE of our monotone estimator over that of the smoothing estimator. Right panel: The relative MSE of our monotone estimator over that of the sample splitting estimator. The solid gray horizontal line represents a ratio of 1.}
    \label{Fig: relative mse}
\end{figure}

\clearpage

\subsection*{Proof of Theorems}

\begin{proof}[\bfseries{Proof of Theorem~\ref{thm:invariant}}]
(1) For any $S$, we can take $\mu = F_S$, so that $f_S=1$ on the support of $S$, and  $\lambda_S = 1/ \bar{F}_{S-}$. We then have $\theta = 1$ on the support of $S$, which is monotone.

(2) Suppose $F_S \geq_{MHR} F_T$ and $F_T \geq_{MHR} F_U$. Let $\s{S} := \n{Supp}(F_S)$, $\s{T} := \n{Supp}(F_T)$, and $\s{U} := \n{Supp}(U)$. We can take $\mu$ to be a measure dominating $F_S$, $F_T$, and $F_U$. We want to show that $\lambda_S(x) /\lambda_U(x) \leq \lambda_S(y) / \lambda_U(y)$ for all $x, y \in \s{S} \cup \s{U}$ such that $x < y$. Let $x,y \in (\s{S} \cup \s{T}) \cap (\s{S} \cup \s{U}) \cap (\s{T} \cup \s{U})$ be such that $x \leq y$. Then since $F_S \geq_{MHR} F_T$ and $x,y \in \s{S} \cup \s{T}$, $\lambda_S(x)/\lambda_T(x) \leq \lambda_S(y)/\lambda_T(y)$. Similarly, since $F_T \geq_{MHR} F_U$ and $x,y \in \s{T} \cup \s{U}$, $\lambda_T(x)/\lambda_U(x) \leq \lambda_T(y)/\lambda_U(y)$. Hence,
\[ \frac{\lambda_S(x)}{\lambda_U(x)} = \frac{\lambda_S(x)}{\lambda_T(x)} \frac{\lambda_T(x)}{\lambda_U(x)} \leq \frac{\lambda_S(y)}{\lambda_T(y)} \frac{\lambda_T(y)}{\lambda_U(y)} =\frac{\lambda_S(y)}{\lambda_U(y)} .\]
We now need to address 
\[x,y \in (\s{S} \cup \s{U}) \backslash [(\s{S} \cup \s{T}) \cap (\s{S} \cup \s{U}) \cap (\s{T} \cap \s{U})] = [\s{S} \backslash (\s{T} \cup \s{U})] \cup  [\s{U} \backslash (\s{S} \cup \s{T})].\]
If $x  \in  \s{U} \backslash (\s{S} \cup \s{T})$, then $\lambda_S(x) / \lambda_U(x) = 0$, which is guaranteed to be no larger than $\lambda_S(y) / \lambda_U(y)$. Similarly, if $y \in   \s{S} \backslash (\s{T} \cup \s{U})$, then $\lambda_S(y) / \lambda_U(y) = +\infty$, which is guaranteed to be no smaller than $\lambda_S(x) / \lambda_U(x)$.

The only remaining case is $x \in  \s{S} \backslash (\s{T} \cup \s{U})$ and $y \in  \s{U} \backslash (\s{S} \cup \s{T})$. We show there cannot simultaneously be such $x,y$ with $x < y$. First, if $x \in \s{S} \backslash (\s{T} \cup \s{U})$, then $\lambda_S(x) / \lambda_T(x) = +\infty$, which implies since  $S \geq_{MHR} T$ that $\lambda_S(z) / \lambda_T(z) = +\infty$ for all $z > x$ with $z \in \s{S} \cup \s{T}$. This implies that $\lambda_T(z)  = 0$ for all such $z$, so that $z$ is not in $\s{T}$. Therefore, $t < x$ for all $t \in \s{T}$. Similarly, if $y \in  \s{U} \backslash (\s{S} \cup \s{T})$, then $\lambda_T(y) / \lambda_U(y)= 0$, which implies since $T \geq_{MHR} U$ that $\lambda_T(z) / \lambda_U(z) = 0$ for all $z < y$ with $z \in \s{T} \cup \s{U}$. This implies that $\lambda_T(z) = 0$ for all such $z$, so that $z$ is not in $\s{T}$. Therefore, $t > y$ for all $t \in \s{T}$. Since $\s{T}$ cannot be empty, this completes the proof of (2).

(3) If $F_S$ and $F_T$ are both dominated by $\mu$, then $F_{\psi(S)} = F_S \circ \psi^{-1}$ and $F_{\psi(T)} = F_T \circ \psi^{-1}$ are both dominated by $\mu \circ \psi^{-1}$. Let $f_{\psi(S)}$ be the density of $F_{\psi(S)}$ with respect to  $\mu \circ \psi^{-1}$, and $F_{\psi(T)}$ be the density of $F_{\psi(T)}$ with respect to  $\mu \circ \psi^{-1}$. Then $f_{\psi(S)} / f_{\psi(T)} = (f_S \circ \psi^{-1}) / (f_T \circ \psi^{-1})$, so 
\[ \frac{\lambda_{\psi(S)}}{\lambda_{\psi(T)}} = \frac{f_{\psi(S)} /\bar{F}_{\psi(S),-}}{f_{\psi(T)} / \bar{F}_{\psi(T),-}} = \frac{f_{\psi(S)}}{f_{\psi(T)}} \frac{\bar{F}_{\psi(T),-}}{\bar{F}_{\psi(S),-}} =\frac{f_{S} \circ \psi^{-1}}{f_{T} \circ \psi^{-1}} \frac{\bar{F}_{T,-} \circ \psi^{-1}}{\bar{F}_{S,-} \circ \psi^{-1}} = \theta \circ \psi^{-1}. \]
Since $\theta$ and $\psi$ are monotone, so is $\theta \circ \psi^{-1}$.
\end{proof}

\begin{lemma}\label{lemma:monotone_comp}
Suppose that $J \subseteq \d{R}$ is an interval and that $\Phi : J \to \d{R}$ is a non-decreasing and c\`{a}dl\`{a}g function, and $r: J \to \d{R}$ is continuous on $\n{Supp}(\Phi)$. Let $\Gamma(x) := \int_{J \cap (-\infty, x]} r(u) \, d\Phi(u)$. Then $r$ is non-decreasing on $\n{Supp}(\Phi)$ if and only if $\Gamma \circ \Phi^-$ is convex on $\n{Im}(\Phi)$, and if $r$ is non-decreasing on $\n{Supp}(\Phi)$ then $r(x) = \partial_- \n{GCM}_{I}(\Gamma \circ \Phi^-) \circ \Phi(x)$ for all $x \in \n{Supp}(\Phi)$, where $I$ is the smallest interval in $\d{R}$ containing $\n{Im}(\Phi)$.
\end{lemma}
\begin{proof}[\bfseries{Proof of Lemma~\ref{lemma:monotone_comp}}]
We first show that if $r$ is non-decreasing on $\n{Supp}(\Phi)$, then $\Gamma \circ \Phi^-$ is convex on $\n{Im}(\Phi)$. Suppose we have $t, u, v \in \n{Im}(\Phi)$, where $t < u < v$ and $u = \delta t + (1 - \delta) v \text{ for } \delta \in (0,1)$. Defining $R_{\Gamma, \Phi}(x) = \Gamma \circ \Phi^-(x)$, we have
\begin{align*}
    R_{\Gamma, \Phi} (v) - R_{\Gamma, \Phi} (u) 
    &= \int_{\Phi^-(u)}^{\Phi^-(v)} r(x)\, d\Phi(x) \\
    &\geq \left[r \circ \Phi^-(u) \right] \left[\Phi \circ \Phi^-(v) - \Phi\circ\Phi^-(u)\right] \\
    &\geq \left[\int_{\Phi^-(t)}^{\Phi^-(u)} r(x)\, d\Phi(x)\right] \left[\frac{\Phi \circ\Phi^-(v) - \Phi \circ\Phi^-(u)}{\Phi \circ \Phi^-(u) - \Phi \circ \Phi^-(t)}\right] \\
    &= \left[R_{\Gamma, \Phi} (u) - R_{\Gamma, \Phi} (t) \right]\left[ \frac{\Phi \circ \Phi^-(v) - \Phi \circ \Phi^-(u)}{\Phi \circ \Phi^-(u) - \Phi \circ \Phi^-(t)} \right].
\end{align*}
Now, for any $x \in \n{Im}(\Phi)$, $\Phi \circ \Phi^-(x) = x$. Thus, by plugging $u = \delta t + (1 - \delta) v$ into the above inequality, we have $R_{\Gamma, \Phi} (v) - R_{\Gamma, \Phi} (u) \geq \frac{\delta}{1 - \delta}[R_{\Gamma, \Phi} (u) - R_{\Gamma, \Phi} (t)]$, which implies that $R_{\Gamma, \Phi} (u) \leq \delta R_{\Gamma, \Phi} (t) + (1 - \delta)R_{\Gamma, \Phi} (v)$. Therefore, $\Gamma \circ \Phi^-$ is convex on $\n{Im}(\Phi)$.

Next, we show that if $r$ is continuous on $\mathscr{G} := \n{Supp}(\Phi)$, and $R_{\Gamma, \Phi} = \Gamma \circ \Phi^-$ is convex on $\n{Im}(\Phi)$, then $r$ is non-decreasing on $\mathscr{G}$.
The idea is to compare the slopes of chords of $r$ using the convexity of $R_{\Gamma, \Phi}$. Let $x, y \in \mathscr{G}$ be such that $x <y$. Suppose there exist two sequences $\{z_j\}_{j\geq1}$ and $\{w_j\}_{j\geq1}$ such that $\lim_{j\to\infty} s_j = r(x)$ and $\lim_{j \to \infty} t_j = r(y)$, where
\[s_j := \frac{R_{\Gamma, \Phi} \circ \Phi(x) - R_{\Gamma, \Phi} \circ \Phi(z_j)}{\Phi(x) - \Phi(z_j)} \text{ and } t_j := \frac{R_{\Gamma, \Phi}\circ \Phi(y) - R_{\Gamma, \Phi} \circ \Phi(w_j)}{\Phi(y) - \Phi(w_j)}.\] 
If $z_j \leq w_j$ for all $j$ large enough, then convexity of $R_{\Gamma, \Phi}$ on $\n{Im}(\Phi)$ would imply that $s_j \leq t_j$ for all $j$ large enough, and hence $r(x) \leq r(y)$. Hence, if we can find such sequences, we have established the claim. The slopes of the chords of $R_{\Gamma, \Phi}$ depend on the behavior of $\Phi$ near $x$ and $y$, each of which has three cases. We note that since $y \in \n{Supp}(\Phi)$ by assumption, exactly one of the following three situations must hold: (y1) $\Phi(y) > \Phi(y-)$ and there exists $p \in [x, y)$ such that $\Phi(p) = \Phi(y-)$, (y2) $\Phi(y) > \Phi(y-)$ but there's no $p \in [x, y)$ such that $\Phi(p) = \Phi(y-)$, and (y3) $\Phi(y) = \Phi(y-)$. There are three analogous cases for $x$: (x1) $\Phi(x) > \Phi(x-)$ and there exists $q < x$ such that $\Phi(q) = \Phi(x-)$, (x2) $\Phi(x) > \Phi(x-)$ but there's no such $q < x$,  and (x3) $\Phi(x) = \Phi(x-)$. We proceed by defining $\{z_j\}_{j\geq1}$ and $\{w_j\}_{j\geq1}$ in each case.


In case (y1), we let $w_j = p$ for all $j$. We can set $p = \Phi^- \circ \Phi(y-)$ and still have $\Phi(p) = \Phi(y-)$. We know that $p \in [x,y)$; if $p = x$ then $p$ is in $\s{G}$ by assumption, and if $p > x$, then $\Phi(u) < \Phi(p)$ for all $u \in [x,p)$ since $p = \Phi^- \circ \Phi(y-)$, in which case $p \in \s{G}$ as well. Hence, $p \in \s{G}$ necessarily. We therefore have
\[ t_j = \frac{R_{\Gamma,\Phi} \circ \Phi (y) - \Gamma(p)}{\Phi(y) - \Phi(p)} = \frac{r(y) [\Phi(y) - \Phi(p)]}{\Phi(y) - \Phi(p)} = r(y)\]
since $\Phi$ is by assumption flat on $(p, y)$ and has a jump at $y$.

In case (y2), we have $\Phi^- \circ \Phi(y-) = y$. Thus, there exists $\{w_j\}_{j\geq1}$ increasing to $y$ such that $w_j \in (x, y) \cap \mathscr{G}$ for each $j$ and $\Phi(w_j)$ increases to $\Phi(y-)$. Therefore, $\Phi^- \circ \Phi(w_j)$ converges to $\Phi^- \circ \Phi(y-) = y$. Hence $R_{\Gamma, \Phi} \circ \Phi(w_j)$ converges to $R_{\Gamma, \Phi} \circ \Phi(y-)$ since $\Gamma$ possesses left-limits, showing that $t_j$ converges to $[R_{\Gamma,\Phi}\circ \Phi(y) - R_{\Gamma,\Phi}\circ \Phi(y-)]/ [\Phi(y) - \Phi(y-)] = r(y)$. In case (y3), since $y \in \mathscr{G}$, there exists $\{w_j\}_{j\geq1}$ that either (y3a) increases to $y$ and $\Phi(w_j) < \Phi(y)$ for each $j$, or (y3b) decreases to $y$ and $\Phi(w_j) > \Phi(y)$ for each $j$. In case (y3a), since $\Phi \circ \pp{z} = \Phi(z)$ for all $z$, we have
\begin{align*}
    t_j &= \left[\int_{\pp{w_j}}^{\pp{y}} r(u) \, d\Phi(u)\right] /\left[\Phi(y) - \Phi(w_j)\right] \\
        &= r(y) \left[\Phi \circ \pp{y} - \Phi \circ \pp{w_j}\right]/\left[\Phi(y) - \Phi(w_j)\right] \\
        &\qquad + \left[\int_{\pp{w_j}}^{\pp{y}}[r(u) - r(y)] \, d\Phi(u) \right]/\left[\Phi(y) - \Phi(w_j)\right] \\
        &= r(y) + \left[ \int_{\pp{w_j}}^{\pp{y}}[r(u) - r(y)] \, d\Phi(u)\right] / \left[\Phi(y) - \Phi(w_j) \right]\ .
\end{align*}
By continuity of $r$ over $\mathscr{G}$, for any $\epsilon > 0$, we can find $m$ such that $j \geq m$ implies $|r(u) - r(y)| < \epsilon$ for all $u \in [w_j, y] \cap \mathscr{G}$. We then have
\begin{align*}
\left| t_j - r(y) \right| &\leq \left[\int_{\pp{w_j}}^{\pp{y}} \left|r(u) - r(y)\right| \, d\Phi(u) \right] / \left[\Phi(y) - \Phi(w_j)\right] \\
                          &\leq \frac{\epsilon[\Phi \circ \pp{y} - \Phi \circ \pp{w_j}]}{\Phi(y) - \Phi(w_j)} =\epsilon
\end{align*}
for all $j \geq m$, so $\lim_{j \rightarrow \infty} t_j = r(y)$. If (y3b) holds, then a similar argument shows that $\lim_{j \rightarrow \infty} t_j = r(y)$. Applying the same exact reasoning for the three cases for $x$, we see that $s_j$ converges to $r(x)$. 

We have now shown that there exist sequences $z_j$ and $w_j$ such that $s_j$ converges to $r(x)$ and $t_j$ converges to $r(y)$, where $s_j$ and $t_j$ are defined above. Hence, if $z_j \leq w_j$ for all $j$ large enough, then convexity of $R_{\Gamma, \Phi}$ on $\n{Im}(\Phi)$ implies that $s_j \leq t_j$ for all $j$ large enough, and hence $r(x) \leq r(y)$. It is clear that $z_j \leq w_j$ for all 16 pairings of definitions of $z_j$ and $w_j$ implied by cases (x1)--(x3b) and (y1)--(y3b) except for when $z_j$ decreases to $x$ (case x3b) and $w_j = p$ (case y1). In this case, we note that if $x = p$, then $\Phi$ is flat on $[x,y)$, so that case (x3b) cannot hold. Therefore, if $z_j$ decreases to $x$ and $w_j = p$, then $p$ must be strictly larger than $x$, so that $z_j < w_j$ for all $j$ large enough.

Next we prove the second claim of the lemma: if $r$ is continuous and non-decreasing on $\mathscr{G}$, then $r(x) = \partial_- \n{GCM}_{I}(\Gamma \circ \Phi^-) \circ \Phi(x)$ for all $x \in \mathscr{G}$, where $I$ is the smallest interval containing $\n{Im}(\Phi)$. We have proved that $R_{\Gamma, \Phi}$ is convex on $\n{Im}(G)$ under the stated conditions. First, we claim that $\n{GCM}_{I}(\Gamma \circ \Phi^-) = H$, where $H := I \to \d{R}$ has the following form. For any $u \in \n{Im}(\Phi), H(u) := R_{\Gamma, \Phi}(u)$. If $u \in I$ but $u \notin \n{Im}(\Phi)$, then there exists $x \in \d{R}$ and $\lambda \in [0, 1)$ such that $u = \lambda \Phi(x-) + (1 - \lambda) \Phi(x)$. We then define $H(u) = \lambda R_{\Gamma, \Phi}(\Phi(u-)-) + (1 - \lambda)R_{\Gamma, \Phi}(\Phi(u))$. Thus defined, $H$ is the linear interpolation of $R_{\Gamma, \Phi}|_{\n{Im}(\Phi)}$ to all of $I$. In order to show that $H = \n{GCM}_{I}(R_{\Gamma, \Phi})$, we need to show that (a) $H$ is convex, (b) $H \leq R_{\Gamma, \Phi}$ and (c) $H \geq \bar{H}$ for any other convex minorant $\bar{H}$ of $R_{\Gamma, \Phi}$.

For (a), let $u, v \in I$ and $p = \lambda u + (1 - \lambda)v$ for $\lambda \in (0, 1)$. Since $I$ is the smallest interval containing $\n{Im}(\Phi)$, there exist $u_1  \leq u_2 \leq p_1 \leq p_2 \leq v_1 \leq v_2$ that are all elements of $\n{Im}(\Phi)$ and $\lambda_S, \lambda_T, \lambda_3 \in [0, 1]$ such that $u = \lambda_Su_1 + (1-\lambda_S)u_2$, $v = \lambda_Tv_1 + (1-\lambda_T)v_2$, and $p = \lambda_3p_1 + (1-\lambda_3)p_2$, and such that $H(u) = \lambda_SR_{\Gamma, \Phi}(u_1-) + (1-\lambda_S)R_{\Gamma, \Phi}(u_2), H(v) = \lambda_TR_{\Gamma, \Phi}(v_1-) + (1-\lambda_T)R_{\Gamma, \Phi}(v_2)$, and $H(p) =\lambda_3R_{\Gamma, \Phi}(p_1-) + (1-\lambda_3)R_{\Gamma, \Phi}(p_2)$. (If $u \in \n{Im}(\Phi)$, then we set $u_1 = u_2 = u$ and $\lambda_S = 0$. Otherwise, we can find $u_1 < u < u_2$ with $u_1$ and $u_2$ in $\n{Im}(\Phi)$ exist since $I$ is the smallest interval containing $\n{Im}(\Phi)$, and such a $\lambda_S$ exists by the definition of $H$. We define $p_1, p_2, v_1, v_2$ similarly, and we can ensure that the stated ordering is satisfied since $u < p < v$.)

We define the points $U_1 := (u_1, H(u_1)), U_2 := (u_2, H(u_2)), ..., V_2:= (v_2, H(v_2)$. The convexity of $R_{\Gamma, \Phi}$ implies that the (possibly degenerate) line segment $\overline{P_1P_2}$ lies on or below $\overline{U_2V_1}$, which lies on or below  $\overline{U_2V}$, which lies on or below $\overline{UV}$. Since $P$ lies on $\overline{P_1P_2}$ and the point $(p, \lambda H(u) + (1-\lambda) H(v)$ lies on the line $\overline{UV}$, we have $H(p) \leq \lambda H(u) + (1-\lambda) H(v)$. Since $u, v$, and $\lambda$ were arbitrary, this implies that $H$ is convex.

For (b), if $u \in \n{Im}(\Phi)$, by definition $H(u) = R_{\Gamma, \Phi}(u)$. If $u \notin \n{Im}(\Phi)$, then since $u = \lambda \Phi(x-) + (1 - \lambda) \Phi(x)$ for some $x$, we must have $\Phi^-(u) = \Phi^- \circ \Phi(x) = x$. Consequently, by the convexity and continuity of $H$ on $\n{Im}(\Phi)$
\begin{align*}
    R_{\Gamma, \Phi}(u) &= \Gamma \circ \Phi^-(u) = \Gamma \circ \Phi^- \circ \Phi(x) =  R_{\Gamma, \Phi} \circ \Phi(x) \\
    &= \lambda R_{\Gamma, \Phi} \circ \Phi(x) + (1 - \lambda)R_{\Gamma, \Phi} \circ \Phi(x) \\
    &\geq \lambda R_{\Gamma, \Phi} \circ \Phi(x-) + (1 - \lambda)R_{\Gamma, \Phi} \circ \Phi(x) \\
    &= \lambda R_{\Gamma, \Phi}(\Phi(x-)-) + (1 - \lambda)R_{\Gamma, \Phi} \circ \Phi(x) \\
    &= H(u).
\end{align*}

For (c), if $\bar{H}$ is another convex minorant of $R_{\Gamma, \Phi}$, then $H(u) = R_{\Gamma, \Phi}(u) \geq \bar{H}(u)$ for all $u \in \n{Im}(\Phi)$. If $u \notin \n{Im}(\Phi)$, we have $u = \lambda \Phi(x-) + (1 - \lambda) \Phi(x)$ for some $x \in \n{Im}(\Phi)$. By convexity of $\bar{H}$ on $\n{Im}(\Phi)$ and since $\bar{H} \leq R_{\Gamma, \Phi}$ by assumption,
\begin{align*}
   \bar{H}(u) &= \bar{H}(\lambda \Phi(x-) + (1 - \lambda) \Phi(x)) \leq \lambda \bar{H} \circ \Phi(x-) + (1 - \lambda) \bar{H} \circ \Phi(x) \\
   &\leq \lambda R_{\Gamma, \Phi} \circ \Phi(x-) +  (1 - \lambda) R_{\Gamma, \Phi} \circ \Phi(x).
\end{align*}
If $\Phi(x-) \in \n{Im}(\Phi)$, then $R_{\Gamma, \Phi} \circ \Phi(x-) = R_{\Gamma, \Phi}(\Phi(x-)-)$ since $R_{\Gamma,\Phi}$ is continuous on $\n{Im}(\Phi)$, so the above equals
\begin{align*}
   &\lambda R_{\Gamma, \Phi}(\Phi(x-)-) + (1 - \lambda) R_{\Gamma, \Phi} \circ \Phi(x) = H(u).
\end{align*}
If $\Phi(x-) \notin \n{Im}(\Phi)$, then for all $\epsilon > 0$ there exists $z \in (\Phi(x-) - \epsilon, \Phi(x-))$ such that $z \in \n{Im}(G)$, since otherwise $\Phi$ would be flat to the left of $x$ and $\Phi(x-)$ would be in $\n{Im}(\Phi)$. We then have $\bar{H}(u) \leq \lambda(z)R_{\Gamma, \Phi}(z-) + (1 - \lambda(z))R_{\Gamma, \Phi}(\Phi(x))$, where $\lambda(z) \in (0, 1)$ and $\lambda(z) \to \lambda$ as $z \to \Phi(x-)$. Taking the limit of $z \to \Phi(x-)$, we have $\bar{H}(u) \leq \lambda R_{\Gamma, \Phi}(\Phi(x-)-) + (1 - \lambda) R_{\Gamma, \Phi}(\Phi(x)) = H(u)$.

Finally, we show that $r(x) = (\partial_- H) \circ \Phi (x)$ for all $x \in \mathscr{G}$. If $\Phi(x) > \Phi(x-)$, then for all $\lambda \in (0, 1)$ and $u = \lambda \Phi(x-) + (1 - \lambda) \Phi(x)$, we have 
\[H(u) = \lambda R_{\Gamma, \Phi}(\Phi(x-)-) + (1 - \lambda) R_{\Gamma, \Phi}(\Phi(x)) = \lambda \Gamma(x-) + (1 - \lambda) \Gamma(x).\]
Thus $(\partial_- H)(u) = [\Gamma(x) - \Gamma(x-)] / [\Phi(x) - \Phi(x-)] = r(x)$. If $\Phi(x) = \Phi(x-)$, then $H(u) = R_{\Gamma, \Phi}(\Phi(x))$ and it's clear that $(\partial_- H)(u) = r(x)$.
\end{proof}

\begin{proof}[\bfseries{Proof of Theorem~\ref{thm:mhr_characterization}}]
We note that $\theta = (dF_S / dF_T) / (\bar{F}_{S,-} / \bar{F}_{T,-}) = d\Lambda_S / d\Lambda_T$ since $d\Lambda_S(u) = d F_S(u) / \bar{F}_{S,-}(u)$ and $d\Lambda_T(u) = dF_T(u) / \bar{F}_{T,-}(u)$. In addition, $\n{Supp}(F_T) = \n{Supp}(\Lambda_T)$. Therefore, we can write $\Lambda_S(t) = \int_{0}^t \theta(u)\, d\Lambda_T(u)$. Hence, (1) $\iff$ (3) and (b) follow by Lemma~\ref{lemma:monotone_comp} with $\Gamma = \Lambda_S$, $\Phi = \Lambda_T$, and $r = \theta$.

It remains to show (1) $\iff$ (2). By Lemma~\ref{lemma:monotone_comp} with $r = \theta$, and $\Phi = F_T$, $\theta$ is nondecreasing on $\n{Supp}(F_T)$ if and only if
\[ u \mapsto \int_{-\infty}^{F_T^-(u)} \theta \, dF_T\]
is convex on $\n{Im}(F_T)$. Hence, if 
\[\int_{[0,u)} \frac{1-v}{\bar{R}(v)}\, dR_+(v) = \int_{-\infty}^{F_T^-(u)}\theta(v) \, dF_T(v),\]
then (1) $\iff$ (2). We can write 
\[ \int_{[0,u)} \frac{1-v}{\bar{R}(v)}\, dR_+(v) = \int_{[0,u)} \frac{1-v}{\bar{R}(v)}\, dR_+^c(v) + \sum_{v < u} \frac{1-v}{\bar{R}(v)} (\Delta R_+)(v) \]
for $R_+^c(u) = R_+(u) - \sum_{v \leq u}  (\Delta R_+)(v)$ the continuous part of $R_+$. We address the discrete and continuous parts of the integral in turn.

We note that $\Delta R_+(v) > 0$ if and only if $R_+(v) > R(v)$, which implies that $F_T^-(v+) = F_{T,+}^-(v) > F_T^-(v)$, since otherwise $R_+(v) = F_S \circ F_{T,+}^-(v) = F_S \circ F_T^-(v) = R(v)$. Furthermore $F_{T,+}^-(v) > F_T^-(v)$ if and only if $F_T$ is flat on $[F_T^-(v), F_{T,+}^-(v))$, which implies $F_S$ is too. But if  $R_+(v) > R(v)$ then $F_S\circ F_{T,+}^-(v) > F_S \circ F_T^-(v)$, which implies $F_S$, and therefore $F_T$ as well, have jumps at $F_{T,+}^-(v)$. Hence, $R_+(v) > R(v)$ if and only if $F_S$ and $F_T$ are both flat on $[F_T^-(v), F_{T,+}^-(v))$ and have a jump at $F_{T,+}^-(v)$, and $v = F_T \circ F_T^-(v) = F_{T,-} \circ F_{T,+}^-(v)$ and $F_S\circ F_T^-(v) = F_{S,-} \circ F_{T,+}^-(v)$. Therefore, for any $v$ with $\Delta R_+(v) > 0$,
\begin{align*}
\frac{1-v}{\bar{R}(v)} \Delta R_+(v) &= \frac{1-v}{1-F_S\circ F_T^-(v)} \left[ F_S \circ F_{T,+}^-(v) - F_S\circ F_T^-(v)\right] = \frac{1-F_{T,-} \circ F_{T,+}^-(v)}{1 - F_{S,-}\circ F_{T,+}^-(v)} (\Delta F_S) \circ F_{T,+}^-(v) \\
&= \frac{(\Delta F_S) \circ G_{T,+}^-(v)}{\bar{F}_{S,-} \circ F_{T,+}^-(v)} \frac{\bar{F}_{T,-} \circ F_{T,+}^-(v)}{(\Delta F_T) \circ F_{T,+}^-(v)} (\Delta F_T) \circ F_{T,+}^-(v) = \frac{\lambda_S \circ F_{T,+}^-(v)}{\lambda_T \circ F_{T,+}^-(v)} (\Delta F_T) \circ F_{T,+}^-(v).
\end{align*}
Therefore,
\begin{align*}
    \sum_{v < u} \frac{1-v}{\bar{R}(v)} \Delta R_+(v) &= \sum_{v < u}  \frac{\lambda_S \circ F_{T,+}^-(v)}{\lambda_T \circ F_{T,+}^-(v)} (\Delta F_T) \circ F_{T,+}^-(v) = \sum_{t < F_{T,+}^-(u)}  \frac{\lambda_S(t)}{\lambda_T(t)} \Delta F_T(t) \\
    &= \sum_{t \leq F_T^-(u)}  \frac{\lambda_S(t)}{\lambda_T(t)} \Delta F_T(t) ,
\end{align*}  
where the last equality follows because $t < F_{T,+}^-(u)$ if and only if $t \leq F_T^-(u)$.

We now address the continuous part of the integral. Using the fact derived above that $F_S \circ F_T^-(v) = F_{S-} \circ F_{T+}^-(v)$ for any $v$ with $\Delta R_+(v) > 0$, we have
\begin{align*}
    R_+^c(u) &= R_+(u) - \sum_{v\leq u} (\Delta R_+)(v) = F_S \circ G_{T,+}^-(v) - \sum_{v \leq u} \left[ F_S \circ G_{T,+}^-(u) - F_S \circ F_T^-(u) \right] \\
    &=  F_S \circ F_{T,+}^-(v) - \sum_{v \leq u} \left[ F_S \circ F_{T,+}^-(u) - F_{S,-} \circ F_{T,+}^-(u) \right] =  F_S \circ F_{T,+}^-(v) - \sum_{v \leq u} (\Delta F_S) \circ  F_{T,+}^-(v) \\
    &= F_S^c \circ F_{T,+}^-(u).
\end{align*}
Therefore,
\[ \int_{[0,u)} \frac{1-v}{\bar{R}(v)}\, dR_+^c(v) = \int_{[0,u)} \frac{1-v}{\bar{F}_S \circ F_T^-(v)}\, (F_S^c \circ F_{T,+}^-)(dv).\]
We then note that $v \in \n{Supp}(R_+^c)$ implies that $v \in \n{Supp}(F_{T,+}^-)$, and hence $v = F_T \circ F_{T,+}^-(v)$ unless $v$ is at the left end of a flat of $F_{T,+}^-$. Such points form a $R_+^c$ measure zero set. Similarly, if $\Delta R_+(v) = 0$, then $F_{T,+}^-(v) = F_T^-(v)$, and $v$ such that $\Delta R_+(v) > 0$ form a $R_+^c$ measure zero set. Therefore, we have
\[ \int_{[0,u)} \frac{1-v}{\bar{F}_S \circ F_T^-(v)}\, (F_S^c \circ F_{T,+}^-)(dv) = \int_{[0,u)} \frac{\bar{F}_T \circ F_{T,+}^-(v)}{\bar{F}_S \circ F_{T,+}^-(v)}\, (F_S^c \circ F_{T,+}^-)(dv).\]
Now we note that $F_{T,+}^-$ is strictly increasing on the support of $R_+^c$, so by the change of variables $y = F_{T,+}^-(v)$, we have
\begin{align*}
    \int_{[0,u)} \frac{1-v}{\bar{R}(v)}\, dR_+^c(v) &= \int_{[0,u)}  \frac{\bar{F}_T \circ F_{T,+}^-(v)}{\bar{F}_S \circ F_{T,+}^-(v)}\, (F_S^c \circ F_{T,+}^-)(dv)  = \int_{[F_{T,+}^-(0),F_{T,+}^-(u))} \frac{\bar{F}_T(y)}{\bar{F}_S(y)}\, F_S^c(dy) \\
    &= \int_{[F_{T,+}^-(0),F_{T,+}^-(u))}\frac{\bar{F}_T(y)}{\bar{F}_S(y)}\, \frac{dF_S^c}{dF_T^c}(y) \, dF_T^c(y).
\end{align*}
Now except on a $F_T^c$-measure zero set, $\bar{F}_T = \bar{F}_{T,-}$, $\bar{F}_S = \bar{F}_{S,-}$, and $\frac{dF_S^c}{dF_T^c} = \frac{dF}{dG}$. Therefore, 
\[ \int_{[0,u)} \frac{1-v}{\bar{R}(v)}\, dR_+^c(v)  = \int_{[F_{T,+}^-(0),F_{T,+}^-(u))}\frac{\bar{F}_{T,-}(y)}{\bar{F}_{S,-}(y)}\, \frac{dF_S}{dF_T}(y) \, dF_T^c(y) = \int_{[F_{T,+}^-(0),F_{T,+}^-(u))}\frac{ d\Lambda_S}{d\Lambda_T}(y) \, dF_T^c(y). \]
Finally, we note that if $F_{T,+}^-(0) > 0$, then $F_T^c$ is flat on $(-\infty,F_{T,+}^-(0))$, and $y < F_{T,+}^-(u)$ if and only if $y \leq F_T^-(u)$, so
\[\int_{[F_{T,+}^-(0),F_{T,+}^-(u))}\frac{ d\Lambda_S}{d\Lambda_T}(y) \, dF_T^c(y) = \int_{-\infty}^{F_T^-(u)}\frac{ d\Lambda_S}{d\Lambda_T}(y) \, dF_T^c(y).\]
Putting together the discrete and continuous parts of the integral, we now have 
\begin{align*}
    \int_{[0,u)} \frac{1-v}{\bar{R}(v)}\, dR_+(v) &= \int_{-\infty}^{F_T^-(u)}\frac{ d\Lambda_S}{d\Lambda_T}(y) \, dF_T^c(y) + \sum_{t \leq F_T^-(u)}  \frac{\lambda_S(t)}{\lambda_T(t)} \Delta F_T(t) = \int_{-\infty}^{F_T^-(u)}\frac{ d\Lambda_S}{d\Lambda_T}(y) \, dF_T(y).
\end{align*}
\end{proof}



We denote $\d{P}_n$ as the empirical distribution of $O_1, \dotsc, O_n$, $P$ as the true distribution of $O_i$ (as implied by $F_S$, $F_T$, $F_U$, $F_V$, and $\pi$), and $\d{G}_n = n^{1/2}\left( \d{P}_n - P\right)$. For any probability distribution $Q$ and $Q$-integrable function $h$, we denote $Qh := \int h \, dQ$. We also let $\pi_n := \sum_{i=1}^n A_i / n$ be the observed fraction of treated units.

\begin{proof}[\bfseries{Proof of Theorem~\ref{thm: convergence in distribution}}]

To prove Theorem~\ref{thm: convergence in distribution}, we will use Theorem~4 of \cite{westling2020unified}. For convenience, we refer to \cite{westling2020unified} as WC hereafter. In the notation of WC, we have $\Gamma_n = \Lambda_{S,n}$, $\Phi_n = \Lambda_{T,n}$, $\Gamma_0 = \Lambda_S$, and $\Phi_0 = \Lambda_T$. To use Theorem~4 of WC, we need to first verify that the decomposition of equation (2) of WC holds, and then verify conditions (B1) -- (B5) and (A4) -- (A5) of WC. We establish each of these in turn below.

{\bf Equation (2) of WC.} 
We define the influence function $D_x^*$ and $L_x^*$ of $\Lambda_{S,n}(x)$ and $\Lambda_{T,n}(x)$ as 
 \begin{align*}
        D_{x}^*(y, \delta, a) &= \frac{a}{\pi}\left[ \frac{I(y \le x, \delta = 1)}{\bar{F}_S(y-)\bar{F}_U(y-)} - \int_0^{x \wedge y} \frac{d\Lambda_{S}(v)}{\bar{F}_S(v-)\bar{F}_U(v-)} \right] \ , \text{ and}\\
         L_{x}^*(y, \delta, a) &= \frac{1-a}{1-\pi}\left[ \frac{I(y \le x, \delta = 1)}{\bar{F}_T(y-)\bar{F}_V(y-)} - \int_0^{x \wedge y} \frac{d\Lambda_{T}(v)}{\bar{F}_T(v-)\bar{F}_V(v-)} \right].
\end{align*}
By adding and subtracting terms, we have $\Lambda_{S,n}(x) - \Lambda_{S}(x) = \d{P}_n D_{x}^* + H_{x,n}$ and $\Lambda_{T,n}(x) - \Lambda_{T}(x) = \d{P}_n L_{x}^* + R_{x,n}$, where 
\begin{align*}
        H_{x,n} &= \Lambda_{S,n}(x) -\Lambda_{S}(x) - \d{P}_n D_{x}^* \ , \text{ and} \\
        R_{x,n} &= \Lambda_{T,n}(x) -\Lambda_{T}(x) - \d{P}_n L_{x}^*.
    \end{align*}
These are the functions corresponding to equation
(2) of WC.

{\bf Condition (B1).} We define the local difference function $g_{x, u}$
\[g_{x, u} := D_{x+u}^* - D_{x}^* - \theta_0(x) \left[ L_{x+u}^* - L_{x}^* \right].\]
To verify condition (B1), we need to bound the uniform entropy of the class $\{ g_{x,u} : |u| \leq R\}$ for all $R$ small enough. We further decompose $g_{x,u} = g_{x,u, S} - \theta_0(x)g_{x,u,T}$ for $g_{x,u,S} := D_{x+u}^* - D_{x}^*$ and $g_{x,u,T} := L_{x+u}^* - L_{x}^*$. The function $g_{x, u, S}$ can be written as
\[
    g_{x, u, S} (y, \delta, a) = \frac{a}{\pi}\left[\frac{\left \{ I(y \leq x + u) - I(y \leq x)\right\}\delta }{\bar{F}_{S}(y-) \bar{F}_{V}(y-)} -\int \frac{I(v \leq y)  \left \{ I(v \leq x + u) - I(v \leq x)\right\} }{\bar{F}_{S}(v-) \bar{F}_{V}(v-)}d\Lambda_{S}(v)\right].
\]
The class of functions $\s{C}_{x,R} := \{ y \mapsto I(y \leq x + u) : |u| \le R\}$ is  Vapnik-\u{C}ervonenkis (VC) with index 2 (see, e.g.\ Lemma~2.6.16 of \citealp{van1996weak}). The class
\[ \left\{ (y,\delta, a) \mapsto  \frac{a\left \{ I(y \leq x + u) - I(y \leq x)\right\}\delta }{\pi\bar{F}_{S}(y-) \bar{F}_{V}(y-)}: |u| \le R\right\}\]
is a Lipschitz transformation of $\s{C}_{x,R}$ and various fixed square-integrable functions, so it is also VC, and hence easily satisfies condition (B1b) of WC. In conjunction with Lemma 5.2 of \cite{van2006estimating}, this also implies that the class 
\[ \left\{ (y,\delta, a) \mapsto  \frac{a}{\pi}\int \frac{I(v \leq y)  \left \{ I(v \leq x + u) - I(v \leq x)\right\} }{\bar{F}_{S}(v-) \bar{F}_{V}(v-)}d\Lambda_{S}(v): |u| \le R\right\}\]
is VC. Hence, $\{ g_{x,u,S} : |u| \leq R\}$ satisfies (B1b). By the analogous forms of $g_{x,u,S}$ and $g_{x,u,T}$, an identical argument shows that $\{ g_{x,u,T} : |u| \leq R\}$ satisfies (B1b), so $\{g_{x,u} : |u| \leq R\}$ does as well.

{\bf Condition (B2).}  An envelope function for $\{ g_{x,u,S} : |u| \leq R\}$ is given by $\hat{G}_{x, R, S} = \hat{G}_{x, R, S, 1} + \hat{G}_{x, R, S, 2}$ for 
\begin{align*}
    \hat{G}_{x, R, S,1}(y, \delta, a) &=\frac{I(|y - x| \leq R) a \delta}{\pi\bar{F}_{S}(y-) \bar{F}_{V}(y-)}\\
    \hat{G}_{x, R, S,2}(y, \delta, a) &= \frac{a}{\pi} \int_{x-R}^{x+R} I(0 < v \leq y) \frac{d\Lambda_{S}(v)}{\bar{F}_{S}(v-) \bar{F}_{V}(v-)}.
\end{align*}
We now verify that $P_0\hat{G}_{x, R, S}^2 = \boundeddet(R)$ as $R \rightarrow 0$ under the stated conditions. 
Due to the boundedness of $\bar{F}_{S}$ and $\bar{F}_{V}$ away from zero and independent censoring,
\begin{align*}
    P_0\hat{G}_{x, R, S,1}^2 &\le  C P(| Y - x| \le R, \Delta = 1 \mid A = 1) \\
    &= C\int_{x - R}^{x + R} \bar{F}_U(u-) \, dF_S(u) \\
    &\leq C'[F_S(x + R) - F_S(x - R)]
\end{align*}
for some $C, C' < \infty$. This latter expression is $\boundeddet(R)$ as $R \rightarrow 0$ because $F_S$ is continuously differentiable at $x$ with finite derivative.

Next, by Jensen's inequality, 
\begin{align*}
    P_0 \hat{G}_{x, R, S,2}^2  &= \pi^{-1} E \left[ \left\{ \int_{x-R}^{x+R} I(0 < v \leq Y) \frac{dF_{S}(v)}{\bar{F}_{S}^2(v-) \bar{F}_{V}(v-)}\right\}^2 \mid A = 1\right] \\
    &\le \pi^{-1} E \left[  \int_{x-R}^{x+R} I(0 < v \leq Y) \frac{dF_{S}(v)}{\bar{F}_{S}^4(v-) \bar{F}_{V}^2(v-)}  \mid A = 1\right] \\
    &= \pi^{-1}  \int_{x-R}^{x+R} P(Y \geq v \mid A= 1) \frac{dF_{S}(v)}{\bar{F}_{S}^4(v-) \bar{F}_{V}^2(v-)}  \\
    &= \pi^{-1}  \int_{x-R}^{x+R} \frac{dF_{S}(v)}{\bar{F}_{S}^3(v-) \bar{F}_{V}(v-)}.
\end{align*}
As above, the last expression is $\boundeddet(R)$ as $R \rightarrow 0$ because $F_S$ is Lipschitz at $x$ and $\bar{F}_S$ and $\bar{F}_U$ are both positive in a neighborhood of $x$. By the triangle inequality, we then have $P_0 \hat{G}_{x, R, S}^2  = \boundeddet(R)$.

For the second part of condition (B2), we note that since $\bar{F}_S$ and $\bar{F}_U$ are both positive in a neighborhood of $x$, $\hat{G}_{x, R, S}$ is uniformly bounded for all $R$ small enough. Hence, for all $\eta$, for all $R$ small enough (possibly depending on $\eta$), $\{R\hat{G}_{x, R, S} > \eta\}$ is identically 0, which implies in particular that $P_0 \left(\hat{G}_{x, R, S}^2\{R\hat{G}_{x, R, S} > \eta\}\right) = \fasterthandet(R)$. Identical analysis applies to an envelope for $\{ g_{x,u,T} : |u| \leq R\}$.

{\bf Condition (B3).} This condition concerns properties of the covariance function defined as $\Sigma(s, t) := P[D_{s}^* - \theta(x)L_{s}^*][D_{t}^* - \theta(x)L_{t}^*]$. Since $a(1-a) = 0$ for $a \in \{0,1\}$, $D_s^* L_t^* = 0$ for any $s,t$. Hence, $\Sigma(s, t) = P[ D_s^* D_t^*] +\theta(x)^2 P [ L_s^* L_t^*]$. We write
\begin{align}
    P\left[D_{s}^* D_{t}^* \right] &= \pi^{-1} E\left[\frac{I(Y \leq s \wedge t, \Delta = 1)}{\bar{F}_S(Y-)^2 \bar{F}_U(Y-)^2} \mid A = 1\right] \nonumber\\
    &\qquad - \pi^{-1} E \left[ \frac{I(Y \leq s, \Delta = 1)}{\bar{F}_S(Y-) \bar{F}_U(Y-)} \int_0^{t \wedge Y} \frac{d\Lambda_S(v)}{\bar{F}_S(v-) \bar{F}_U(v-)}  \mid A = 1 \right] \nonumber \\
    &\qquad -\pi^{-1} E \left[ \frac{I(Y \leq t, \Delta = 1)}{\bar{F}_S(Y-) \bar{F}_U(Y-)} \int_0^{s \wedge Y} \frac{d\Lambda_S(v)}{\bar{F}_S(v-) \bar{F}_U(v-)}   \mid A = 1 \right] \nonumber \\
    &\qquad + \pi^{-1} E \left[  \int_0^{s \wedge Y} \frac{d\Lambda_S(v)}{\bar{F}_S(v-) \bar{F}_U(v-)}\int_0^{t \wedge Y} \frac{d\Lambda_S(v) }{\bar{F}_S(v-) \bar{F}_U(v-)} \mid A = 1 \right].\label{eq:PDsDt}
\end{align}
We will address each term in this expansion. We first have
\begin{align*}
    E \left[\frac{ I(Y \leq s \wedge t, \Delta = 1)}{\bar{F}_S(Y-)^2 \bar{F}_U(Y-)^2} \mid A = 1\right] &=  \int_0^{s\wedge t} \frac{ d\Lambda_S(v)}{\bar{F}_S(v-) \bar{F}_U(v-)}.
\end{align*}
Next, the second term can be simplified as follows
\begin{align*}
  E\left[ \frac{I(Y \leq s, \Delta = 1)}{\bar{F}_S(Y-) \bar{F}_U(Y-)} \int_0^{t \wedge Y} \frac{d\Lambda_S(v)}{\bar{F}_S(v-) \bar{F}_U(v-)}  \mid A = 1 \right] &= \int_0^s \int_0^{t \wedge y} \frac{d\Lambda_S(v)}{\bar{F}_S(v-) \bar{F}_U(v-)}  \, d\Lambda_S(y) \\
  &= \int_0^{s \wedge t} \int_v^{s} \, d\Lambda_S(y) \frac{d\Lambda_S(v)}{\bar{F}_S(v-) \bar{F}_U(v-)}   \\
  &=\int_0^{s \wedge t} \frac{\Lambda_S(s) - \Lambda_S(v)}{\bar{F}_S(v-) \bar{F}_U(v-)} \, d\Lambda_S(v).
\end{align*}
Similarly, the third term can be written as
\begin{align*}
  E \left[ \frac{I(Y \leq t, \Delta = 1)}{\bar{F}_S(Y-) \bar{F}_U(Y-)} \int_0^{s \wedge Y} \frac{d\Lambda_S(v)}{\bar{F}_S(v-) \bar{F}_U(v-)}   \mid A = 1 \right]  &= \int_0^{s \wedge t} \frac{\Lambda_S(t) - \Lambda_S(v)}{\bar{F}_S(v-) \bar{F}_U(v-)} \, d\Lambda_S(v).
\end{align*}
Finally, for the last term, by decomposing the double integral into two regions, we can write
\begin{align*}
    &E\left[  \int_0^{s \wedge Y} \frac{d\Lambda_S(v)}{\bar{F}_S(v-) \bar{F}_U(v-)}\int_0^{t \wedge Y} \frac{d\Lambda_S(v) }{\bar{F}_S(v-) \bar{F}_U(v-)} \mid A = 1 \right] \\
    &\qquad= \int_0^s \int_0^t \frac{\bar{F}_S(u \vee v -) \bar{F}_U(u \vee v -)}{\bar{F}_S(u-) \bar{F}_U(u-) \bar{F}_S(v-) \bar{F}_U(v-)} \, d\Lambda_S(u)d\Lambda_S(v) \\
    &\qquad= \int_0^{s \wedge t} \frac{\Lambda_S(t) -\Lambda_S(v)}{\bar{F}_S(v-) \bar{F}_U(v-)} \, d\Lambda_S(v) + \int_0^{s \wedge t} \frac{\Lambda_S(s) -\Lambda_S(v)}{\bar{F}_S(v-) \bar{F}_U(v-)} \, d\Lambda_S(v).
\end{align*}
Hence, the second through fourth terms in the decomposition of $P[ D_s^* D_t^*]$ provided in equation~\eqref{eq:PDsDt} cancel, and we are left with
\[ P[ D_s^* D_t^* ]= \pi^{-1} \int_0^{s\wedge t} \frac{ d\Lambda_S(v)}{\bar{F}_S(v-) \bar{F}_U(v-)}.\]
By symmetry of $D_s^*$ and $L_s^*$, we also have 
\[ P[ L_s^* L_t^* ] = (1-\pi)^{-1} \int_0^{s\wedge t} \frac{ d\Lambda_T(v)}{\bar{F}_T(v-) \bar{F}_V(v-)}.\]
Thus, we can write 
\[ \Sigma(s,t) = \int_0^{s\wedge t}\left[\frac{\lambda_S(v)}{\pi\bar{F}_S(v-) \bar{F}_U(v-)} +\frac{ \theta(x)^2\lambda_T(v)}{(1-\pi) \bar{F}_T(v-) \bar{F}_V(v-)}  \right] \, dv. \]
In the notation of condition (B3) of WC, we have $\Sigma^*(s,t) = 0$, 
\[A(s,t,v,w) = \frac{\lambda_S(v)}{\pi\bar{F}_S(v-) \bar{F}_U(v-)} +\frac{ \theta(x)^2\lambda_T(v)}{(1-\pi) \bar{F}_T(v-) \bar{F}_V(v-)},\]
$H(v,w) = v$, and $Q$ can be taken as any probability measure since there are no covariates.

Sub-conditions (B3a) and (B3d) are automatically satisfied. Sub-condition (B3b) is satisfied because $A$ does not depend on $s$ or $t$. Sub-condition (B3c) requires that $v \mapsto A(x,x,v,w)$ be continuous at $v = x$, which would appear to require that $F_U$ and $F_V$ are continuous at $x$.  However, in the proof of Theorem~4 on pages 4 and 5 of the Supplementary Material of WC, it is actually only used that $v \mapsto A(x,x,v,w)$ possesses a right-limit $A(x,x,x+,w)$ as $v$ approaches $x$ from above (since $\alpha$ in the proof stands in for $n^{-1/3}$, which is positive), in which case $A_0(x,x,x+,w)$ should take the place of $A(x,x,x,w)$ in the result.  This is important in our work because in many applications the censoring distributions $F_U$ and $F_V$ possess mass points. Hence, this weaker version of sub-condition (B3c) holds with 
\begin{align*}
    A(x,x,x+,w) &= \frac{\lambda_S(x)}{\pi\bar{F}_S(x) \bar{F}_U(x)} +\frac{ \theta(x)^2\lambda_T(x)}{(1-\pi) \bar{F}_T(x) \bar{F}_V(x)} \\
    &= \theta(x) \left[ \frac{\lambda_T(x)}{\pi\bar{F}_S(x) \bar{F}_U(x)} + \frac{\lambda_S(x)}{(1-\pi)\bar{F}_T(x) \bar{F}_V(x)}\right].
\end{align*} 
Thus, condition (B3) holds, and the scale parameter is
\begin{align*}
    \kappa(x) &= \int A(x,x,x+, w) H'(x, w) Q(dw) =\theta(x) \left[ \frac{\lambda_T(x)}{\pi\bar{F}_S(x) \bar{F}_U(x)} + \frac{\lambda_S(x)}{(1-\pi)\bar{F}_T(x) \bar{F}_V(x)}\right].
\end{align*}

{\bf Conditions (B4) and (B5).}  For these two conditions, we define $\tilde{H}_{u, n} := H_{x+u, n} - H_{x, n}$, $\tilde{R}_{u, n} := R_{x+u, n} - R_{x, n}$ and $K_n(\delta) := n^{2/3}\sup_{|u| \leq \delta n^{-1/3}}\left|\tilde{H}_{u, n} - \theta(x)\tilde{R}_{u, n}\right|$. For (B4), we need to show that $K_n(\delta) = \fasterthan(1)$ for each $\delta > 0$, and for (B5), we need to show that for some $\alpha \in (1,2)$, $\delta \mapsto \delta^{-\alpha} E[ K_n(\delta)]$ is decreasing for all $n$ large enough and $\delta$ small enough. As above, we only verify the conditions for $\tilde{H}_{u,n}$, since verification for $\tilde{R}_{u,n}$ is completely analogous.

We define
\begin{align*}
    \tilde{F}_S(x) &:= P(Y \leq x, \Delta = 1 \mid A = 1)\ ,\\ \tilde{R}_S(x) &:= P(Y \geq x \mid A = 1) \ ,\\
    \tilde{F}_{S, n}(x) &= \d{P}_n(Y \leq x, \Delta = 1 \mid A = 1) = \frac{1}{n\pi_n}\sum_{i=1}^n I(Y_i \leq x, \Delta_i = 1, A_i = 1) \ , \\
    \tilde{R}_{S, n}(x) &:= \d{P}_n(Y \geq x \mid A = 1) = \frac{1}{n\pi_n} \sum_{i=1}^n I(Y_i \geq x, A_i =1).
\end{align*}
Then by the definition of the Nelson-Aalen estimator and the definition of the cumulative hazard function with independent right censoring, we can write
\begin{align*}
    H_{x,n} &= \Lambda_{S, n}(x) - \Lambda_S(x) - \d{P}_nD_x^* \\
    &= \int_0^x \frac{d\tilde{F}_{S, n}}{\tilde{R}_{S, n}} - \int_0^x \frac{d\tilde{F}_S}{\tilde{R}_S} - \int\left\{\frac{a}{\pi}\left[\frac{I(y \leq x, \delta = 1)}{\tilde{R}_{S}(y)} - \int_0^{x \wedge y}\frac{d\tilde{F}_S}{\tilde{R}_S^2}\right]\right\} \, d\d{P}_n(y,\delta, a)  \\
    &= \int_0^x\left[ \frac{d\tilde{F}_{S, n}}{\tilde{R}_{S, n}} - \frac{d\tilde{F}_S}{\tilde{R}_S} - \frac{\pi_n}{\pi}\frac{d\tilde{F}_{S, n}}{\tilde{R}_{S}} + \frac{\pi_n}{\pi} \frac{\tilde{R}_{S, n}}{\tilde{R}_S^2} \, d\tilde{F}_S\right] \\
    &= \int_0^x \left[\frac{1}{\tilde{R}_{S, n}} - \frac{\pi_n}{\pi\tilde{R}_{S}}\right]\, d\tilde{F}_{S, n}  - \int_0^x \left[\frac{1}{\tilde{R}_S} - \frac{\pi_n}{\pi} \frac{\tilde{R}_{S, n}}{\tilde{R}_S^2} \right] \, d\tilde{F}_S \\
    &= \int_0^x \left[\frac{\pi\tilde{R}_{S} - \pi_n\tilde{R}_{S, n}}{\pi\tilde{R}_{S}\tilde{R}_{S, n}}\right]\, d(\tilde{F}_{S, n} - \tilde{F}_S) - \int_0^x \left[\frac{\pi\tilde{R}_{S} - \pi_n\tilde{R}_{S, n}}{\pi\tilde{R}_{S}\tilde{R}_{S, n}}\right] \left[\frac{\tilde{R}_{S, n}}{\tilde{R}_S} - 1\right] \, d\tilde{F}_S\\
    &= \int_0^x \left[\frac{\pi\tilde{R}_{S} - \pi_n\tilde{R}_{S, n}}{\pi\tilde{R}_{S}\tilde{R}_{S, n}}\right]\, d(\tilde{F}_{S, n} - \tilde{F}_S) + \int_0^x \frac{\left[\pi_n\tilde{R}_{S, n} - \pi\tilde{R}_{S} \right]^2}{\pi^2\tilde{R}_{S}^2\tilde{R}_{S, n}} \, d\tilde{F}_S \\
    &\qquad+ \left[ \pi_n^{-1} - \pi^{-1} \right] \int_0^x \frac{\pi_n\left[\pi_n\tilde{R}_{S, n} - \pi\tilde{R}_{S} \right]}{\pi \tilde{R}_S^2} \, d\tilde{F}_S.
\end{align*}
We let $a_{n,\delta} := x-\delta n^{-1/3}$ and $b_{n,\delta} := x+\delta n^{-1/3}$. We note that $\pi > 0$ and that $\tilde{R}_S$ and $\tilde{R}_{S,n}$ are bounded away from zero in a neighborhood of $x$ almost surely for all $n$ large enough. Hence, there is a constant C such that almost surely for all $n$ large enough and $\delta$ small enough,
\begin{align*}
    &E\left[ \sup_{|u| \leq \delta n^{-1/3}} |H_{x+u, n} - H_{x, n}| \right] \\
    &\qquad= E\left[ \sup_{|u| \leq \delta n^{-1/3}} \left| \int_x^{x+u} \left[\frac{\pi\tilde{R}_{S} - \pi_n\tilde{R}_{S, n}}{\pi\tilde{R}_{S}\tilde{R}_{S, n}}\right]\, d(\tilde{F}_{S, n} - \tilde{F}_S)+ \int_x^{x+u} \frac{\left[\pi\tilde{R}_{S} - \pi_n\tilde{R}_{S, n}\right]^2}{\pi^2\tilde{R}_{S}^2\tilde{R}_{S, n}} \, d\tilde{F}_S   \right. \right. \\
    &\qquad\qquad \left.\left.+\left[ \pi_n^{-1} - \pi^{-1} \right] \int_x^{x+u} \frac{\pi_n\left[\pi_n\tilde{R}_{S, n} - \pi\tilde{R}_{S} \right]}{\pi \tilde{R}_S^2} \, d\tilde{F}_S \right| \right] \\
    &\qquad\leq E\left[ \int_{a_{n,\delta}}^{b_{n,\delta}} \left|\frac{\pi\tilde{R}_{S} - \pi_n\tilde{R}_{S, n}}{\pi\tilde{R}_{S}\tilde{R}_{S, n}}\right|\, d|\tilde{F}_{S, n} - \tilde{F}_S| \right]+  E\left[\int_{a_{n,\delta}}^{b_{n,\delta}} \frac{\left[\pi\tilde{R}_{S} - \pi_n\tilde{R}_{S, n}\right]^2}{\pi^2\tilde{R}_{S}^2\tilde{R}_{S, n}} \, d\tilde{F}_S \right] \\
    &\qquad\qquad +  E\left[\left| \pi_n^{-1} - \pi^{-1} \right| \int_{a_{n,\delta}}^{b_{n,\delta}} \frac{\pi_n\left|\pi_n\tilde{R}_{S, n} - \pi\tilde{R}_{S} \right|}{\pi \tilde{R}_S^2} \, d\tilde{F}_S \right] \\
    &\qquad\leq C\left\{ E\left[ \left\|\pi\tilde{R}_{S} - \pi_n\tilde{R}_{S, n} \right\|_{\infty, [a_{n,\delta}, b_{n,\delta}]} \left\| \tilde{F}_{S, n} - \tilde{F}_S \right\|_{TV, [a_{n,\delta}, b_{n,\delta}]} \right]\right. \\
    &\qquad\qquad+\left.  E\left[\left\|\pi\tilde{R}_{S} - \pi_n\tilde{R}_{S, n} \right\|_{\infty, [a_{n,\delta}, b_{n,\delta}]}^2\right]   +  E\left[\left| \pi_n^{-1} - \pi^{-1} \right| \left\|\pi\tilde{R}_{S} - \pi_n\tilde{R}_{S, n} \right\|_{\infty, [a_{n,\delta}, b_{n,\delta}]}\right]\right\}\\
    &\qquad\leq C\left\{ E\left[ \left\|\pi\tilde{R}_{S} - \pi_n\tilde{R}_{S, n} \right\|_{\infty, [a_{n,\delta}, b_{n,\delta}]}^2 \right] E \left[\left\| \tilde{F}_{S, n} - \tilde{F}_S \right\|_{TV, [a_{n,\delta}, b_{n,\delta}]}^2 \right]\right\}^{1/2} \\
    &\qquad\qquad+ CE\left[\left\|\pi\tilde{R}_{S} - \pi_n\tilde{R}_{S, n} \right\|_{\infty, [a_{n,\delta}, b_{n,\delta}]}^2\right]   +  C\left\{E\left[\left| \pi_n^{-1} - \pi^{-1} \right|^2\right] E\left[ \left\|\pi\tilde{R}_{S} - \pi_n\tilde{R}_{S, n} \right\|_{\infty, [a_{n,\delta}, b_{n,\delta}]}^2\right]\right\}^{1/2}.
\end{align*}
We address each term in turn. For any fixed $x$ and $\eta \geq 0$, we define the function class $\tilde{\mathcal{F}}_{S, \eta} := \{ (y, a) \mapsto I(y \geq u, a = 1) : u \in [x - \eta, x + \eta]\}$. The class $\tilde{\mathcal{F}}_{S,\eta}$ is uniformly bounded by 1 and $P_0$-Donsker for any $\eta$. Since $\pi\tilde{R}_{S}(x) = P(Y \geq x, A = 1)$ and $\pi_n\tilde{R}_{S, n}(x) = \d{P}_n(Y \geq x, A = 1)$, we can then write 
\[\left\{E \left[ \left\| \pi_n \tilde{R}_{S,n} - \pi \tilde{R}_S \right\|_{\infty, [a_{n,\delta}, b_{n,\delta}]}^2\right]\right\}^{1/2} = n^{-1/2} \left\{E \left[\left\{ \sup_{f \in \tilde{\mathcal{F}}_{S, \delta n^{-1/3}}} \left| \d{G}_n f  \right| \right\}^2\right]\right\}^{1/2} = \boundeddet(n^{-1/2}) .\]
Therefore, we also have
\[ E \left[ \left\| \pi_n \tilde{R}_{S,n} - \pi \tilde{R}_S \right\|_{\infty, [a_{n,\delta}, b_{n,\delta}]}^2\right] = \boundeddet(n^{-1}).\]
Next, since $\tilde{F}_S$ and $\tilde{F}_{S,n}$ are non-decreasing functions, we have
\begin{align*}
    \left\| \tilde{F}_{S, n} - \tilde{F}_S \right\|_{TV, [a_{n,\delta}, b_{n,\delta}]} &\leq \left\| \tilde{F}_{S, n}\right\|_{TV, [a_{n,\delta}, b_{n,\delta}]} + \left\| \tilde{F}_S \right\|_{TV, [a_{n,\delta}, b_{n,\delta}]} \\
    &= \left[\tilde{F}_{S, n}(b_{n,\delta}) - \tilde{F}_{S, n}(a_{n,\delta})\right] + \left[\tilde{F}_S(b_{n,\delta}) - \tilde{F}_S(a_{n,\delta})\right] \\
    &= \left[\tilde{F}_{S, n}(b_{n,\delta}) - \tilde{F}_S(b_{n,\delta})\right] - \left[\tilde{F}_{S, n}(a_{n,\delta}) - \tilde{F}_S(a_{n,\delta})\right] + 2\left[\tilde{F}_S(b_{n,\delta}) - \tilde{F}_S(a_{n,\delta})\right] \\
    &\leq 2 \left\|\tilde{F}_{S, n} - \tilde{F}_S\right\|_{\infty, [a_{n,\delta}, b_{n, \delta}]} + 2\left[\tilde{F}_S(b_{n,\delta}) - \tilde{F}_S(a_{n,\delta})\right].
\end{align*}
Using a similar approach as for $ \pi_n \tilde{R}_{S,n} - \pi \tilde{R}_S$, we can show that 
\[ \left\{E \left[\left\| \pi_n\tilde{F}_{S, n}-\pi\tilde{F}_S\right\|_{\infty, [a_{n,\delta}, b_{n, \delta}]}^2\right]\right\}^{1/2} = \boundeddet(n^{-1/2}).\]
Turning to $\tilde{F}_S(b_{n,\delta}) - \tilde{F}_S(a_{n,\delta})$, since $\tilde{F}_S(x) = \int_0^x F_U(t-) \, dF_S(t)$ and $F_S$ is continuously differentiable in a neighborhood of $x$, we have 
\[\tilde{F}_S(b_{n,\delta}) - \tilde{F}_S(a_{n,\delta}) = \int_{x-\delta n^{-1/3}}^{x+\delta n^{-1/3}} F_U(t-)\, dF_S(t) \leq F_S(x+\delta n^{-1/3}) - F_S(x-\delta n^{-1/3}) = \boundeddet{(\delta n^{-1/3})}\ .
\]
We conclude that $\left\{E\left[\left\| \tilde{F}_{S, n} - \tilde{F}_S \right\|_{TV, [a_n, b_n]}^2 \right]\right\}^{1/2} = \boundeddet{(\delta n^{-1/3} +n^{-1/2})}$.
Finally, since $\pi > 0$, we have $E|1/\pi_n - 1/\pi| = \boundeddet(n^{-1/2})$. Putting it together, we have 
\[
    E \left[n^{2/3}\sup_{|u| \leq \delta n^{-1/3}} |H_{x+u, n} - H_{x, n}|\right] = n^{2/3} \boundeddet{( n^{-1} + \delta n^{-5/6})} = \boundeddet{(n^{-1/3} + \delta n^{-1/6})}.
\]
This goes to zero for each $\delta > 0$, which verifies (B4), and (B5) is satisfied for any $\alpha \in (1,2)$.

{\bf Condition (A4).} For condition (A4), it suffices to show that $E [ \sup_{t\leq x +\delta}| \Lambda_{T,n}(t) - \Lambda_{T}(t)|] = \fasterthandet(n^{-1/3})$ for some $\delta > 0$. We define $\tilde{F}_T$, $\tilde{F}_{T,n}$, $\tilde{R}_T$, and $\tilde{R}_{T,n}$ as we did above for $S$, but with $A = 0$ in the conditionals instead. We then have
\begin{align*}
 | \Lambda_{T,n}(t) - \Lambda_{T}(t)| &= \left| \int_0^t \frac{d\tilde{F}_{T,n}}{\tilde{R}_{T,n}} - \int_0^t \frac{d\tilde{F}_{T}}{\tilde{R}_{T}} \right| \leq \left| \int_0^t \frac{d(\tilde{F}_{T,n} - \tilde{F}_{T})}{\tilde{R}_{T,n}} \right| +\left|\int_0^t \frac{\tilde{R}_T - \tilde{R}_{T,n}}{\tilde{R}_{T}\tilde{R}_{T,n}} \, d\tilde{F}_{T} \right| \\
 &\leq \left\| \tilde{F}_{T} - \tilde{F}_{T,n} \right\|_{\infty, [0,t]} \left[ 2 / \tilde{R}_{T,n}(t) - 1\right] + \left\| \frac{\tilde{R}_{T} - \tilde{R}_{T,n}}{\tilde{R}_T\tilde{R}_{T,n}} \right\|_{\infty, [0,t]}  \tilde{F}_T(t).
\end{align*}
We used integration by parts to bound the first term in the second inequality. By assumption, $\tilde{R}_T$ is bounded away from zero in a neighborhood of $x$, and as a result, $\tilde{R}_{T,n}$ is almost surely bounded away from zero in a neighborhood of $x$ for all $n$ large enough. Then, for some $\delta > 0$ and $C >0$, almost surely for all $n$ large enough it holds that
\begin{align*}
    E \left[ \sup_{t\leq x +\delta}| \Lambda_{T,n}(t) - \Lambda_{T}(t)| \right]  &\leq CE\left[\left\| \tilde{F}_{T} - \tilde{F}_{T,n} \right\|_{\infty, [0,\delta]}+ \left\|\tilde{R}_{T} - \tilde{R}_{T,n} \right\|_{\infty, [0,\delta]}\right].
\end{align*}
We can show that this expression is $\boundeddet(n^{-1/2})$ using similar empirical process techniques as we did with $S$ above.


{\bf Condition (A5).}  For this condition, since $I_n \subset [0,\gamma_n]$, it suffices to show that the stratified Nelson-Aalen estimators are uniformly consistent on $[0,\gamma_n]$, i.e.\ $\|\Lambda_{T,n} - \Lambda_{T}\|_{\infty, [0,\gamma_n]}$  and $\|\Lambda_{T,n} - \Lambda_{T}\|_{\infty, [0,\gamma_n]}$ tend to zero in probability. This follows from Corollary~1.2 of \cite{stute1994strong} by the assumed lower bound for $r_n$. 
\end{proof}
\end{adjustwidth}
\end{document}